\documentclass[12pt,a4paper]{article}

\usepackage{amsmath, amsfonts}
\usepackage{amssymb, graphicx}
\usepackage{amsthm}
\usepackage{array}
\usepackage[english]{babel}
\usepackage{color}
\usepackage{hyperref}
\usepackage{epsfig}
\usepackage{amsmath}
\usepackage{natbib}
\usepackage{nccmath}
\usepackage{enumerate}
\usepackage{setspace}
\usepackage{pgf,tikz}
\usepackage{mathrsfs}
\usetikzlibrary{arrows}
\usepackage{doi}
\hypersetup{colorlinks,
            linkcolor=blue,
            anchorcolor=blue,
            citecolor=blue,urlcolor={blue}}

\newtheorem{theorem}{Theorem}

\newtheorem{proposition}[theorem]{Proposition}
\theoremstyle{definition}
\newtheorem{definition}[theorem]{Definition}
\newtheorem{example}[theorem]{Example}
\newtheorem{remark}[theorem]{Remark}

\newtheorem{corollary}[theorem]{Corollary}

\definecolor{darkblue}{rgb}{0.0,0.0,0.3}
\definecolor{ballblue}{rgb}{0.13, 0.67, 0.8}
\definecolor{bleudefrance}{rgb}{0.19, 0.55, 0.91}
\definecolor{celestialblue}{rgb}{0.29, 0.59, 0.82}
\definecolor{dartmouthgreen}{rgb}{0.05, 0.5, 0.06}
\definecolor{dukeblue}{rgb}{0.0, 0.0, 0.61}
\definecolor{bostonuniversityred}{rgb}{0.8, 0.0, 0.0}
\definecolor{darkseagreen}{rgb}{0.56, 0.74, 0.56}
\definecolor{lightseagreen}{rgb}{0.13, 0.7, 0.67}
\definecolor{mediumseagreen}{rgb}{0.24, 0.7, 0.44}
\definecolor{mediumelectricblue}{rgb}{0.01, 0.31, 0.59}
\hypersetup{colorlinks,
            linkcolor=dukeblue,
            anchorcolor=dukeblue,
            citecolor=bostonuniversityred,urlcolor={dukeblue}}

\begin{document}

\title{A many-to-one job market: more about the core and the competitive salaries\thanks{A. Atay is a Serra H\'{u}nter Fellow. A. Atay and M. N\'{u}\~{n}ez gratefully acknowledge financial support by the Spanish Ministerio de Ciencia e Innovaci\'{o}n through grant PID2020-113110GB-100/AEI/10.130339/501100011033, by the Generalitat de Catalunya through grant 2021-SGR-00306. T. Solymosi gratefully acknowledges financial support from the Hungarian National Research, Development and Innovation Office via the grant NKFI K-146649. This work supported by National Science Foundation under Grant No DMS-1928930 while Ata Atay was in residence at the Mathematical Science Research Institute in Berkeley, California, during the Fall 2023 semester.}}
\author{
Ata Atay\thanks{
Departament de Matem\`{a}tica Econ\`{o}mica, Financera i Actuarial, and Barcelona Economics Analysis Team (BEAT), Universitat de Barcelona, Spain. E-mail: \href{mailto:aatay@ub.edu}{aatay@ub.edu}}
\and 
Marina N\'{u}\~{n}ez\thanks{Departament de Matem\`{a}tica Econ\`{o}mica, Financera i Actuarial, and Barcelona Economics Analysis Team (BEAT), Universitat de Barcelona, Spain. E-mail: \href{mailto:mnunez@ub.edu}{mnunez@ub.edu}}
\and 
Tam\'{a}s Solymosi \thanks{Corvinus Center for Operations Research, and Department of Operations Research and Actuarial Sciences, Corvinus University of Budapest. E-mail: \href{mailto:tamas.solymosi@uni-corvinus.hu}{tamas.solymosi@uni-corvinus.hu}}}

\date{\today}
\maketitle
\begin{abstract}
This paper studies many-to-one assignment markets, or matching markets with wages. Although it is well-known that the core of this model is non-empty, the structure of the core has not been fully investigated. To the known dissimilarities with the one-to-one assignment game, we add that the bargaining set does not coincide with the core and the kernel may not be included in the core. 
Besides, not all extreme core allocations can be obtained by means of a lexicographic maximization or a lexicographic minimization procedure, as it is the case in the one-to-one assignment game.

The maximum and minimum competitive salaries are characterized in two ways: axiomatically and by means of easily verifiable properties of an associated 
directed graph.
Regarding the remaining extreme core allocations of the many-to-one assignment game,  we propose a lexicographic procedure that, for each order on the set of workers, sequentially maximizes or minimizes each worker's competitive salary. This procedure provides all extreme vectors of competitive salaries, that is all extreme core allocations.
\vspace{0.15cm} \vspace{0.15cm}\vspace{0.15cm}
\end{abstract}

\section{Introduction}

We consider a two-sided market, one side formed by a set of firms and the other by a set of workers. Firms want to hire many workers, up to each firm's capacity, but each worker can work for only one firm. Each firm places a non-negative value on each worker, and workers may have a reservation value. Since we assume firms value groups of workers additively, the main data of the market is the value that each firm-worker pair can attain when matched. This value can be transferred by means of the salary the firm pays to each worker it hires.
A matching is an assignment of a group of workers to each firm and a coalitional game is introduced where the worth of a coalition is the highest value that can be obtained by matching firms and workers in the coalition without violating the capacity of each firm. A natural solution concept in this setting is the core, which is the set of allocations of the total value of the market that cannot be improved upon by any coalition. 

The many-to-one assignment market is an extension of the well-known one-to-one \emph{assignment game} introduced by \cite{ss71} to study two-sided markets where there are indivisible goods which are traded between sellers and buyers in exchange for money. 
In their model,  each buyer wants at most one unit of good, and each seller owns exactly one indivisible good.\footnote{\citep{nr15} is a survey on assignment markets and games.} 
In the setting of two-sided assignment markets where each agent has a unit capacity, \cite{ss71} show that the core is always non-empty and each core element is supported by competitive prices (or salaries). Furthermore, the core has a complete lattice structure. Hence, in our job-market setting,  there exists a unique firm-optimal (worker-optimal) core allocation such that the payoff of every firm (worker) is at least as good as under any other core allocation. Moreover, there is an opposition of interest between the two sides of the market when comparing two core allocations: all agents in one side agree on which of the two they prefer. As a consequence, the best core allocation for one side is the worst for the other side of the market. \cite{d82} and \cite{l83} prove that in the firm-optimal core allocation each firm attains its marginal payoff and in the worker-optimal core allocation each worker attains her marginal payoff. As a consequence of that, the optimal stable rules,  that given an assignment market select the core allocation that is optimal for one side of the market, cannot be manipulated by the agents of that side.

One-to-one assignment market games can be extended to many-to-many assignment markets in two different ways. The first one is known as the multiple-partners assignment game \citep{s92} and each agent can establish several partnerships, as many as its capacity allows, but each of them with a different partner. In the second extension, sometimes known as the transportation game \citep{ssetal01}, an agent may establish more than one partnership with a same agent of the opposite side. In both cases the core is proved to be  non-empty. Notice that our many-to-one situation lies in the intersection of these two extensions since the unitary capacity of agents on one side rules out the possibility of more than one partnership between a same firm-worker pair.


Most of the existing results for one-to-one assignment markets cannot be extended to the many-to-many assignment markets (\citealp{so02}; \citealp{s07}): some core allocations may not be supported by competitive prices and 
the core may not be a lattice. Nevertheless, for many-to-one assignment markets the core has a lattice structure based on the partial order on the set of payoffs to the side of the market where agents have unitary capacity. This guarantees in that case the existence of an optimal core allocation for each side of the market,  However, the firm-optimal stable rule may no longer be non-manipulable by the firms.

We can consider two types of many-to-one assignment markets, depending on which side of the market has unitary capacity. This does not affect the core, but makes a big difference when competitive equilibria are considered. If the unitary capacity is on the side that posts prices, that is sellers or workers, then we are in the many-to-one model of \cite{so02} and the core coincides with the set of competitive equilibrium payoff vectors. When the unitary capacity is on the agents that report a demand given some prices, that is, buyers or firms, then we are in the many-to-one model of \cite{k76}, and the core may strictly contain the set of competitive equilibrium payoff vectors (CE payoff vectors) that now coincides with the set of solutions of the dual linear program that finds an optimal matching. In this paper we focus on the first case, the job market with unitary capacity workers of \cite{so02}, and only in the last section we show that, after some adjustments, parallel results can be obtained for Kaneko's buyer-seller market, where buyers have unitary capacity.



In the first part of the paper, we focus on the core of the many-to-one assignment market games and consider other related set-solution concepts different from the core to show that more dissimilarities appear with respect to the one-to-one case.  We prove that, differently from the one-to-one assignment market, the kernel may not be a subset of the core for many-to-one assignment markets. Therefore, the core and the classical bargaining set do not coincide (Example \ref{ex:kernel} and Corollary \ref{cor:kernel}). The loss of this coincidence means that the core is somehow a less robust solution in these markets, since those allocations that do not have a justified objection could be taken into account when looking for a distribution of the worth of the grand coalition.


However, if we want the allocation of the total worth of the market to be supported by competitive prices (or salaries in our case), then we must restrict our attention again to the set of core payoffs, that coincides with the set of competitive equilibrium payoffs. Generically this set contains infinitely many payoff vectors, although little is known about its dimension and geometric structure.
Applications such as some auction markets, usually consider only the  maximum competitive prices (salaries) rule and the minimum  competitive prices (salaries) rule. We begin by providing axiomatic characterizations of these rules on the domain of many-to-one assignment markets, based on those for the optimal stable rules of the general (many-to-many) multiple-partners assignment market in \citep{dn22}. This is complemented with another characterization based on easily verifiable properties of an associated directed graph, called the tight digraph (more details below). 


Beyond these two extreme core allocations or salary vectors, there are usually many other, and they give an idea of the extension of the core and hence the possibilities of cooperative agreements. To this end, we aim to study the extreme core allocations. First, we observe that, unlike the one-to-one case, in an extreme core allocation it may be the case that no agent achieves his/her marginal contribution  and moreover extreme core allocations are not obtained by a lexicographic minimization procedure or by lexicographic maximization procedure as it is the case for one-to-one assignment markets. See  \citep{ietal07} and \citep{ns17}.  

Based on the projection of the core to the space of workers' payoffs (salaries), and given a competitive salary vector, we define a digraph, the tight digraph, where the set of nodes is the set of workers augmented by a node representing their outside option and the directed arcs are determined by the constraints of the set of competitive salaries that are tight at that given vector. Then, we show that a competitive salary vector is an extreme point if and only if the base-graph of the tight digraph (where the direction of the arcs are ignored) is connected (Theorem \ref{thm:ext_core_char} (A)). It implies that at an extreme competitive salary vector  
there is a worker with zero salary or a worker with a salary that equals the total surplus it creates with a firm under an optimal matching. We also provide a necessary and sufficient condition for each side-optimal allocation in terms of the tight digraph (Theorem \ref{thm:ext_core_char} (B)). 

After that, for each order on the set of workers, we define a payoff vector where each worker sequentially maximizes or minimizes its competitive salary, preserving what has been allocated to its predecessors. Making use of the tight digraph, we show that this set of max-min vectors includes all the extreme competitive salary vectors of the many-to-one assignment market. This gives a procedure for the computation of these extreme points and consequently allows for a representation of the entire core.

Besides the one-to-one assignment game, the literature contains results regarding the set of extreme core allocations for other related combinatorial models, such as ordinal two-sided markets \citep{bb00,bb02} and minimum cost spanning tree games \citep{tvp17}. 

Before concluding, we move to the other many-to-one assignment market, let us say a buyer-seller market where buyers have unitary demand. Our results on the core trivially apply to this case, simply focusing on the projection of the core to the buyers' payoffs. We also provide a description of the competitive equilibrium payoff vectors that allows for a characterization of their extreme points by means of an extended tight graph.

The paper is organized as follows. In Section \ref{sec:prel}, some preliminaries on transferable utility games are provided. Section \ref{sec:model} introduces many-to-one assignment markets and games. In Section \ref{sec:core_geo}, we introduce our results on the core, the kernel, and the bargaining set. We provide an axiomatic characterization of the maximum and minimum competitive salary vectors in Section \ref{sec:max-min}. Another characterization, in terms of properties of associated tight digraphs, is given in Section \ref{extreme_vectors_digraph} as a special case for the characterization of any extreme competitive salary vector.
Based on this, we describe a lexicographic procedure to obtain all extreme vectors of competitive salaries, or extreme core allocations, in Section \ref{sec:extreme_char}.
We consider a special subclass of many-to-one assignment markets and provide some positive results in Section \ref{dominant_diagonal_markets}. Section \ref{sec:kaneko} extends our previous results to the reverse many-to-one model in \citep{k76}, where buyers have unitary demand, and Section \ref{sec:remarks} concludes.

\section{Notations and definitions}
\label{sec:prel}
A \textit{transferable utility (TU) cooperative game} $(N,v)$ is a pair where $N$ is a non-empty, finite set of \textit{players (or agents)} and $v:2^{N}\rightarrow\mathbb{R}$ is a \textit{coalitional function} satisfying $v(\emptyset)=0$.  The number $v(S)$ is regarded as the worth of the coalition $S\subseteq N$. We identify the game with its coalitional function since the player set $N$ is fixed throughout the paper. The game $(N,v)$ is called \textit{superadditive} if $S\cap T=\emptyset$ implies $v(S\cup T)\geq v(S)+v(T)$ for every two coalitions $S,T\subseteq N$. 
Coalition $R\subseteq N$ is called \emph{inessential} in game $v$ if it has a nontrivial partition $R=S\cup T$ with $S,T\neq\emptyset$ and $S\cap T=\emptyset$ such that $v(R)\leq v(S)+v(T)$. Notice that in a superadditive game the weak majorization can only happen as equality. Those non-empty coalitions which are not inessential are called \emph{essential}. Note that the single-player coalitions are essential in any game, and any inessential coalitional value can be weakly majorized by the value of a partition composed only of essential coalitions.  

Given a game $(N,v)$, a \textit{payoff allocation} $x\in\mathbb{R}^{N}$ represents the payoffs to the players. The total payoff to coalition $S\subseteq N$ is denoted by $x(S)=\sum_{i\in S} x_{i}$, in particular $x(\emptyset)=0$, for throughout the paper we keep the convention that summing over the empty-set gives zero. In a game $v$, we say the payoff allocation $x$ is \textit{efficient}, if $x(N)=v(N)$. The set of \emph{imputations}, denoted by $I(v)$, consists of all efficient payoff vectors that are \emph{individually rational}, that is, $x_{i}\ge v(\{i\})$ for all $i\in N$. The core $C(v)$ is the set of imputations that are \emph{coalitionally rational}, that is, $x(S)\ge v(S)$ for all $S\subseteq N$. Observe that all the coalitional rationality conditions for inessential coalitions are implied by the inequalities related to essential coalitions, hence can be ignored: the core and the essential-core are always the same.     

Given a game $(N,v)$, the game $(N, v^*)$ defined by $v^* (S)=v(N)-v(N\setminus S)$ for all $S\subseteq N$ is called the {\em dual game}. 
Notice that $v^*(\emptyset)=0$ and $v^*(N)=v(N)$ for any game $(N,v)$.
It is easily seen that the core of any coalitional game coincides with the \emph{anticore} of its dual game, that is, 
\begin{equation}\label{anticore}
\mathbf{C}(v)=\mathbf{C}^*(v^*):=\{x\in\mathbb{R}^N :\, x(N)=v^*(N), \;x(S)\leq v^*(S) \; \forall S\subseteq N\}. 
\end{equation}
It follows that if $i\in N$ is a \emph{null player} in game $v$ (i.e. $v(S\cup \{i\})=v(S)$ for all $S\subseteq N$, in particular, $v(\{i\})=0$), its payoff is $x_i=0$ in any core allocation $x\in C(v)$. Indeed, then $v(N)=v(N\setminus \{i\})+v(\{i\})\leq x(N\setminus \{i\})+x_i=x(N)=v(N)$, implying both $x(N\setminus \{i\})= v(N\setminus \{i\})$ and $x_i=v(i)=0$.

An order on the set of players $N$ is a bijection $\sigma  :\{1, 2, \ldots , n\}\rightarrow N$, where for all $i\in\{1,2,\ldots,n\}$, $\sigma_{i}=\sigma(i)$ is the player that occupies position $i$. For a given order $\sigma$, $P^{\sigma}_{i}=\{ j \in N\mid \sigma^{-1}(j)<\sigma^{-1}(i)\}$ denotes the set of predecessors of agent $i\in N$. For each order $\sigma$ on the player set $N$ of game $(N,v)$, a \emph{marginal payoff vector} $m^{\sigma, v}$ is defined by $m^{\sigma,v}_{\sigma_i}=v(P^{\sigma}_{\sigma_{i}}\cup\{\sigma_{i}\})-v(P^{\sigma}_{\sigma_{i}})$ for all $i\in N$. Whenever a marginal payoff vector is in the core, then it is an extreme core allocation.

\cite{hetal02} showed that each extreme core allocation of an assignment game is a marginal payoff vector. Nevertheless, the opposite implication only holds in convex assignment games.  

\cite{ns17} studied other  lexicographic allocation procedures for coalitional games looking for a characterization and a computation procedure of their extreme core points. Given a game $(N,v)$ and over the set $\mathbf{Ra}^*(N,v)=\{x\in\mathbb{R}^{N}:x(S)\leq v(N)-v(N\setminus S) \text{ for all } S\subseteq N\}$ of dual coalitionally rational payoff vectors, the following lexicographic maximization procedure is proposed:  
for any order $\sigma$ of the players, the {\em $\sigma$-lemaral} vector
$\overline{r}^{\sigma,v}\in\mathbb{R}^N$ is defined by, for all $i\in\{1,2,\ldots,n\}$,
\begin{equation}\label{lemaral-max}
\overline{r}_{\sigma_i}^{\sigma,v}=\max\left\{x_{\sigma_i}\;:\; x\in \mathbf{Ra}^*(N,v) ,\, x_{\sigma_l}=\overline{r}_{\sigma_l}^{\sigma,v}\;\forall l\in\{1,\ldots,i-1\} \right\},
\end{equation}
which trivially leads to
\begin{equation}\label{lemaral}
\overline{r}_{\sigma_i}^{\sigma,v} =  
\min\left\{v^*(Q\cup\{\sigma_i\})-\overline{r}^{\sigma,v}(Q) \; :\; Q\subseteq P_{\sigma_i}^{\sigma}\right\}.
\end{equation}
 Notice that in the $\sigma$-lemaral vector, the first player in the order maximizes its payoff on the set ${\bf Ra}^*$, the second player maximizes its payoff over those dual coalitional rational payoff vectors that allocate $\overline{r}_{\sigma_1}^{\sigma, v}$ to the first player, and so on. It is proved in \citep{ns17} that the set of extreme core allocations of a one-to-one assignment game coincides with the set of lemaral vectors.

\section{The many-to-one assignment market and game}
\label{sec:model}
We consider a market where there are two types of agents: a finite set of firms $F=\{f_1,f_2,\ldots, f_m\}$ and a finite set of workers $W=\{w_1, w_2 ,\ldots, w_n\}$ where the number of firms $m$ can be different from the number of workers $n$. Let $N=F\cup W$ be the set of all agents. We sometimes denote a generic firm and a generic worker by $i$ and $j$, respectively. 
Each firm $i\in F$ values hiring worker $j\in W$ by $h_{ij}\ge 0$, and each worker $j\in W$ has a reservation value $t_j\ge 0$. Hence, a pair of firm $i\in F$ and worker $j\in W$ can generate a non-negative income $a_{ij}=\max\{h_{ij}-t_j,0\}$. If $h_{ij}\ge t_j$, the income $a_{ij}$ is obtained when firm $i\in F$ hires worker $j\in W$, and it is shared by means of the salary $y_j$ that the firm pays to the worker.  The valuation matrix denoted by $A=(a_{ij})_{(i,j)\in F\times W}$ represents the pairwise income for each possible firm-worker pair. Each firm $i\in F$ would like to hire up to $r_{i}\ge 0$ workers and each worker $j\in W$ can work for at most one firm. Then, as long as we do not perform a strategic analysis, a many-to-one assignment market is the quadruple $\gamma=(F,W,A,r)$. When we analyze whether an agent may profit from misrepresenting its true valuations, we describe the market by $\gamma=(F,W,h,t,r)$.

A \emph{matching} $\mu$ for the market $\gamma=(F,W,A,r)$ is a set of $F\times W$ pairs such that each firm $i\in F$ appears in at most $r_{i}$ pairs and each worker $j\in W$  in at most one pair. We denote by $\mathcal{M}(F,W,r)$ the set of matchings for market $\gamma$. A matching $\mu\in\mathcal{M}(F,W,r)$ is \textit{optimal} for $\gamma$ if $\sum\limits_{(i,j)\in\mu} a_{ij} \geq \sum\limits_{(i,j)\in\mu'} a_{ij}$
holds for any other matching $\mu'\in\mathcal{M}(F,W,r)$. We denote by $\mathcal{M}_{A}(F,W,r)$ the set of optimal matchings for the market $\gamma$. Given a matching $\mu\in\mathcal{M}(F,W,r)$, the set of workers matched to firm $i\in F$ under $\mu$ is $\mu(i)=\{ j\in W\mid (i,j)\in \mu\}$.
It will be convenient to denote the set of workers unmatched under $\mu$ by $\mu(f_0)$, that is $\mu(f_0)=W\setminus \bigcup_{i\in F} \mu(i)$. Observe that $i\neq k\in F$ implies $\mu(i)\cap \mu(k)=\emptyset$, hence $\mu(f_0)\cup \bigcup_{i\in F} \mu(i)=W$ is a partition of the set of workers.

Given a many-to-one assignment market $\gamma=(F,W,A,r)$, we define the income maximization linear programming problem by
\begin{eqnarray}
\label{LP}
\mathcal{V}(F,W)=\max & \sum\limits_{i\in F}\sum\limits_{j\in W} a_{ij}x_{ij}  \\
\text{ s. t. } &  \sum\limits_{j\in W} x_{ij}\leq r_{i}, & i\in F \nonumber \\
&  \sum\limits_{i\in F} x_{ij}\leq 1, & j\in W \nonumber \\
&  x_{ij}\geq 0, & (i,j)\in F\times W. \nonumber   
\end{eqnarray}
It is well known that any variable in any basic feasible solution of this LP problem with integer right hand sides is integral, hence, by the worker capacity inequalities, 0 or 1.
Consequently, the relation $(i,j)\in\mu$ $\leftrightarrow$ $x_{ij}=1$ defines a bijection between the set of basic feasible solutions to this LP problem and the set of matchings $\mu\in\mathcal{M}(F,W,r)$. Henceforth, the optimum value of (\ref{LP}) gives the maximum of the sum of values of the matched pairs while respecting the capacities of firms. 

Given market $\gamma=(F,W,A,r)$, we also apply the above notation and terminology for any submarket $\gamma_{(S,T)}=(S,T,A_{(S,T)},r_S)$ with $S\subseteq F$, $T\subseteq W$, and accordingly restricted payoff matrix $A_{(S,T)}$ and capacity vector $r_S$. 

Now, let us associate a coalitional game with transferable utility (TU-game) with this type of two-sided matching markets. Given a many-to-one assignment market $\gamma=(F,W,A,r)$, its associated \emph{many-to-one assignment game} is the pair $(N,v_{\gamma})$ where $N=F\cup W$ is the set of players and the coalitional function is given by $v_{\gamma}(S\cup T)=\max\limits_{\mu\in\mathcal{M}(S,T,r_S)}\sum\limits_{(i,j)\in\mu}a_{ij}$ for all $S\subseteq F$ and $T\subseteq W$.\footnote{When no confusion arises, for a given market $\gamma$, we denote its corresponding coalitional function by $v$ instead of $v_{\gamma}$.}
For brevity, we denote coalition $S\cup T$ with $S\subseteq F$ and $T\subseteq W$ by $(S,T)$, in particular, one-sided coalitions by $(\emptyset,T)$ and $(S,\emptyset)$. 
As the union of matchings for disjoint coalitions is a matching for the union of the coalitions, i.e. $\mu\in\mathcal{M}(S,T,r_S)$ and $\mu'\in\mathcal{M}(S',T',r_{S'})$ with $S\cap S'=\emptyset$ and $T\cap T'=\emptyset$ implies $\mu\cup\mu'\in\mathcal{M}(S\cup S',T\cup T',r_{S\cup S'})$, it easily follows that many-to-one assignment games are superadditive.
On the other hand, if $\nu\in\mathcal{M}(S,T,r_S)$ is an optimal matching for coalition $(S,T)$, that is $v_{\gamma}(S,T)=\sum\limits_{(i,j)\in\nu}a_{ij}=\sum\limits_{i\in S} \sum\limits_{j\in\nu(i)} a_{ij}$, then it follows from $v_{\gamma}(i,\nu(i))=\sum\limits_{j\in\nu(i)} a_{ij}$ for all $i\in S$ that $v_{\gamma}(S,T)=\sum\limits_{i\in S}v_{\gamma}(i,\nu(i))$.
Since $(S,T)=(\emptyset,\nu(f_0)) \cup \bigcup\limits_{i\in S} (i,\nu(i))$ where $\nu(f_0)$ denotes the unmatched workers in $T$ under $\nu$, and $v_{\gamma}(\nu(f_0))=0=\sum\limits_{j\in\nu(f_0)} v_{\gamma}(j)$,
we get the following observations.
\begin{proposition}
\label{inessentials}
In many-to-one assignment games, the following types of coalitions are inessential:
\begin{itemize}
\item any coalition containing at least two firms, 
\item any single-firm coalition containing more workers than the capacity of the firm,  
\item any one-sided coalition containing at least two players.
\end{itemize}
\end{proposition}
\noindent
Consequently, in an $(m+n)$-player many-to-one assignment game, among the $2^{m+n}-1$ non-empty coalitions, at most $\sum_{i=1}^{m} \sum_{t=1}^{r_i} \binom{n}{t} \leq 2^{n}-2$ coalitions can be essential. However, this exponential upper bound is sharp (if all $a_{ij}$ pairwise income values are positive, $m=2$, and $n=r_1+r_2$).  
We will see in Proposition~\ref{core:description_workers} that the core is already described by a quadratic number of easily identifiable essential coalitions.

As in any coalitional game, the main concern is how to share the worth of the grand coalition (the total income) among all agents. To do so, we focus on the solution concept known as the \textit{core}. Different than one-to-one assignment games, where the only essential coalitions are the individual ones and the mixed-pairs, here (see Proposition \ref{inessentials}), instability may arise from a group of workers and a firm that can be better off by recontracting among themselves instead of their prescribed agreements.

\subsection{Core and competitive salaries}


In order to investigate the relationship between the core of  many-to-one assignment markets and other set-valued solution concepts,  we need to analyze closely the structure of the core of these games. Given a many-to-one assignment market $\gamma=(F,W,A,r)$ and $\mu\in\mathcal{M}_{A}(F,W,r)$ an optimal matching, $(x,y)\in\mathbb{R}^{F}_{+}\times \mathbb{R}^{W}_{+}$ is in the core $C(v_{\gamma})$ of the associated game if and only if for every firm $i\in F$, 
\begin{equation}\label{core:descripton}
x_{i}+\sum\limits_{j\in T}y_{j}\geq \sum\limits_{j\in T} a_{ij}=v_{\gamma}(i,T) \mbox{ for all }T\subseteq W \mbox{ with }|T|\leq r_{i} \mbox{(with equality for $T=\mu(i)$)}
\end{equation} 
and the payoff to unassigned firms or workers is zero.

The above description of the core of a many-to-one assignment game is based on Proposition~\ref{inessentials} and the general equivalence of the core and the essential-core. As we remarked there, this description is still of exponential size (in the number of players). Next, we present a quadratic-size equivalent description of the core, just in terms of the workers' payoffs. It rests on the observation that in the essential-core description (\ref{core:descripton}) only those single-firm coalitions are needed for which $|T\cap\mu(i)|=r_i-1$.

For brevity of exposition, first we balance the model, if needed. In case the total capacity of the firms $\sum_{i\in F} r_i$ exceeds the number of workers $n$, we introduce $\sum_{i\in F} r_i -n>0$ dummy workers who can only generate zero income with any firm. Exclusively from this situation, in case of $n>\sum_{i\in F} r_i$, we introduce a dummy firm, say $f_0$, requiring at most $r_0=n-\sum_{i\in F} r_i >0$ number of workers, but who can only generate zero income with any worker. 
Technically, if needed, we extend matrix $A$ with $\sum_{i\in F} r_i-n>0$ full 0 columns, or with one full 0 row.
This clearly means that we extend the associated many-to-one assignment game $v$ with one or more null players.
Since the core payoff to any null player $j$ is $x_j=0$, the core of the original game is precisely the $x_j=0$ section of the core of the extended game, we can assume without loss of generality that the market $\gamma=(F,W,A,r)$ is \emph{capacity-balanced}, i.e. $n=\sum_{i\in F} r_i$ holds. To keep the exposition simple, we do not introduce any new notation for the possible extended models.

Given a capacity-balanced market $\gamma=(F,W,A,r)$ ($n=\sum_{i\in F} r_i$), we will always assume, due to the non-negativity of matrix $A$ without loss of generality, that in any optimal matching $\mu\in\mathcal{M}_A(F,W,r)$ any firm $i\in F$ is assigned precisely $r_i$ workers ($|\mu(i)|=r_i$), and no worker is unmatched under $\mu$.
For any worker $j\in W$, let $j^\mu\in F$ denote the unique firm for which $j\in\mu(j^\mu)$ holds, that is $j^\mu=\mu^{-1}(j)$.  

The next description of the core in terms of the workers' payoffs follows easily from (\ref{core:descripton}) and is a simplification of the one already given in \citep{so02} for the general, not  necessarily capacity-balanced, market.

\begin{proposition}
\label{core:description_workers}
Given a capacity-balanced many-to-one assignment market $\gamma=(F,W,A,r)$, let $\mu\in\mathcal{M}_{A}(F,W,r)$ be an optimal matching. Then, $(x,y)\in\mathbb{R}^{F}\times \mathbb{R}^{W}$ is in the core of the associated game $C(v_{\gamma})$ if and only if
\begin{enumerate}[(i)]
\item $0\le y_{j}\leq a_{j^\mu j}$ \, for any $j\in W$;
\item $y_k-y_j\ge a_{j^{\mu}k}-a_{j^{\mu}j}$ \, for any $j,k\in W$ such that $j^\mu \neq k^\mu$;
\item $x_{i}=\sum\limits_{j\in\mu(i)}(a_{ij}-y_{j})$ for all $i\in F$.
\end{enumerate}
\end{proposition}
Notice that the number of constraints is $2n=\sum_{i\in F} 2r_i$ in item $(i)$, $\sum_{i\in F} r_i(n-r_i)=n\sum_{i\in F} r_i-\sum_{i\in F} r_i^2 \leq n^2-\sum_{i\in F} r_i = n^2-n$ in item $(ii)$, and $m\leq n$ in item $(iii)$, altogether at most $n^2+2n$.

Also, given any vector of salaries $y\in\mathbb{R}^W$ that satisfies constraints (i) and (ii) above for some optimal matching $\mu$, the payoff to each firm is uniquely determined. Let us denote by $C(W)$ the set of salaries (or wages) that satisfy (i) and (ii), that is, the projection of the core to the space of workers' payoffs. It is proved in \citep{so02} that $C(W)$ is endowed with a lattice structure under the partial order induced by $\mathbb{R}^W$. In fact, the reader will see that constraints (i) and (ii) in Proposition \ref{core:description_workers} define what is named a 45-degree polytope in \citep{q91}.


Here is an illustrative example of the above core description. 
\begin{example}
\label{example:core_workers}
Consider a many-to-one assignment market $\gamma=(F,W,A,r)$ where $F=\{f_{1},f_{2}\}$, $W=\{w_{1},w_{2},w_{3}\}$ are the set of firms and the set of workers respectively, and the capacities of the firms are $r=(2,1)$. The pairwise valuation matrix is the following:
$$A =\bordermatrix{~ & w_{1} & w_{2} & w_{3} \cr
                  f_{1} & 8 & 6 & 3\cr
                  f_{2} & 7 & 6 & 4}.$$
Since the (unique) optimal matching assigns workers $w_1$ and $w_2$ to firm $f_1$, and $w_3$ to firm $f_2$, in any core allocation $x_1=(8+6)-y_1-y_2$ and $x_2=4-y_3$ hold. Henceforth, in terms of the workers payoffs, the core is described by the following system (given in two equivalent forms):
$$
\begin{array}{ccccl}
y_1 & y_2 & y_3 & \geq & 0 \\
\hline
y_1 &     &     & \leq & 8 \\
    & y_2 &     & \leq & 6 \\
    &     & y_3 & \leq & 4 \\
\hline
-y_1 &     & y_3 & \geq & -5=3-8 \\
    & -y_2 & y_3 & \geq & -3=3-6 \\
\hline
y_1 &     & -y_3 & \geq & 3=7-4 \\
    & y_2 & -y_3 & \geq & 2=6-4 \\
\hline 
\end{array}
\qquad\qquad
\begin{array}{ccccccl}
\hline
0 & \leq & y_1 &     &     & \leq & 8 \\
0 & \leq &     & y_2 &     & \leq & 6 \\
0 & \leq &     &     & y_3 & \leq & 4 \\
\hline
3 & \le & y_1 &     & -y_3 & \le & 5 \\
2 & \le &     & y_2 & -y_3 & \le & 3 \\
\hline
\end{array}
$$
Notice the similarities to the one-to-one assignment case, but due to the capacity $r_1=2$ of firm $f_1$, there is no direct relation between the payoffs to its optimally matched workers, two-way direct pairwise comparisions are only between workers assigned to different firms.  

Figure \ref{fig:core} illustrates the $C(W)$ of this example, where the 45-degree lattice structure can be seen.

\definecolor{qqqqff}{rgb}{0.,0.,1.}
\definecolor{aqaqaq}{rgb}{0.6274509803921569,0.6274509803921569,0.6274509803921569}
\definecolor{uuuuuu}{rgb}{0.26666666666666666,0.26666666666666666,0.26666666666666666}
\definecolor{xdxdff}{rgb}{0.49019607843137253,0.49019607843137253,1.}
\definecolor{ududff}{rgb}{0.30196078431372547,0.30196078431372547,1.}
\definecolor{cqcqcq}{rgb}{0.7529411764705882,0.7529411764705882,0.7529411764705882}
\begin{figure}[h]
\centering
\begin{tikzpicture}[scale=0.68][line cap=round,line join=round,>=triangle 45,x=1.0cm,y=1.0cm]
\clip(-4.261547144867606,-5.333072337441284) rectangle (12.421267575684405,3.501482858715046);
\draw [line width=0.8pt] (-3.,-2.)-- (-3.,2.);
\draw [line width=0.8pt] (-3.,-2.)-- (5.014354222068339,-3.607187767800357);
\draw [line width=0.8pt] (-3.,-2.)-- (3.002828217348916,-0.809909417487408);
\draw [line width=0.8pt] (5.014354222068339,-3.607187767800357)-- (5.014354222068339,0.40014919472662075);
\draw [line width=0.8pt] (-3.,2.)-- (5.014354222068339,0.40014919472662075);
\draw [line width=0.8pt] (5.014354222068339,-3.607187767800357)-- (11.017502142402869,-2.3971291555863288);
\draw [line width=0.8pt] (3.002828217348916,-0.809909417487408)-- (11.017502142402869,-2.3971291555863288);
\draw [line width=0.8pt] (3.002828217348916,-0.809909417487408)-- (3.0185432642607863,3.165997451215829);
\draw [line width=0.8pt] (-3.,2.)-- (3.0185432642607863,3.165997451215829);
\draw [line width=0.8pt] (11.017502142402869,-2.3971291555863288)-- (11.001787095490998,1.5473476192931674);
\draw [line width=0.8pt] (3.0185432642607863,3.165997451215829)-- (11.001787095490998,1.5473476192931674);
\draw [line width=0.8pt] (5.014354222068339,0.40014919472662075)-- (11.001787095490998,1.5473476192931674);
\draw [line width=0.8pt,dash pattern=on 2pt off 2pt,color=aqaqaq] (-0.9574827512165562,-1.5950607856191104)-- (3.0185432642607863,3.165997451215829);
\draw [line width=0.8pt,dash pattern=on 2pt off 2pt,color=aqaqaq] (8.5465886982794E-4,-1.4050656224530043)-- (3.0147543200490867,2.2073945656558065);
\draw [line width=0.8pt,dash pattern=on 2pt off 2pt,color=aqaqaq] (3.0147543200490867,2.2073945656558065)-- (11.005793806365977,0.5416631896736419);
\draw [line width=0.8pt,dash pattern=on 2pt off 2pt,color=aqaqaq] (-0.9574827512165562,-1.5950607856191104)-- (7.011384378493428,-3.204645039934567);
\draw [line width=0.8pt,dash pattern=on 2pt off 2pt,color=aqaqaq] (7.011384378493428,-3.204645039934567)--(11.001787095490998,1.5473476192931674);
\draw [line width=0.8pt,dash pattern=on 2pt off 2pt,color=aqaqaq] (8.00104000042622,-3.005159482633978)--(11.005793806365977,0.5416631896736419);
\draw [line width=0.8pt,dash pattern=on 2pt off 2pt,color=aqaqaq] (8.5465886982794E-4,-1.4050656224530043)-- (8.00104000042622,-3.005159482633978);
\draw [line width=0.8pt,dash pattern=on 2pt off 2pt,color=aqaqaq] (0.004198544585103825,-2.6024579172711846)-- (6.020046759883498,-1.4074370112050412);
\draw [line width=0.8pt,dash pattern=on 2pt off 2pt,color=aqaqaq] (6.020046759883498,-1.4074370112050412)-- (8.991963408663445,0.9614459332207967);
\draw [line width=0.8pt,dash pattern=on 2pt off 2pt,color=aqaqaq] (1.999480781125392,-3.002589121236023)-- (8.012410827053339,-1.8020032284288727);
\draw [line width=0.8pt,dash pattern=on 2pt off 2pt,color=aqaqaq] (8.012410827053339,-1.8020032284288727)-- (11.005793806365977,0.5416631896736419);
\draw [line width=2.4pt,,color=qqqqff] (2.033886234509821,-2.1992696825552556)-- (4.016823964518636,-2.5997915253164083);
\draw [line width=2.4pt,,color=qqqqff] (4.016823964518636,-2.5997915253164083)-- (4.996380001776475,-2.404206153365795);
\draw [line width=2.4pt,,color=qqqqff] (2.033886234509821,-2.1992696825552556)-- (10.044295835790313,1.7414846266734245);
\draw [line width=2.4pt,,color=qqqqff] (10.044295835790313,1.7414846266734245)-- (11.001787095490998,1.5473476192931674);
\draw [line width=2.4pt,,color=qqqqff] (11.001787095490998,1.5473476192931674)-- (11.005793806365977,0.5416631896736419);
\draw [line width=2.4pt,,color=qqqqff] (11.005793806365977,0.5416631896736419)-- (4.996380001776475,-2.404206153365795);
\draw [line width=2.4pt,,color=qqqqff] (10.011739140043158,0.38443414781475027)-- (4.016823964518636,-2.5997915253164083);
\draw [line width=2.4pt,,color=qqqqff] (10.011739140043158,0.38443414781475027)-- (11.005793806365977,0.5416631896736419);
\draw [line width=2.4pt,,color=qqqqff] (10.011739140043158,0.38443414781475027)-- (10.044295835790313,1.7414846266734245);
\draw [line width=2.4pt,,color=qqqqff] (10.011739140043158,0.38443414781475027)-- (11.001787095490998,1.5473476192931674);
\draw [line width=1.2pt,,color=qqqqff] (2.033886234509821,-2.1992696825552556)-- (3.006119596912251,-2.006139943894078);
\draw [line width=1.2pt,,color=qqqqff] (3.006119596912251,-2.006139943894078)-- (8.991963408663445,0.9614459332207967);
\draw [line width=1.2pt,,color=qqqqff] (8.991963408663445,0.9614459332207967)-- (10.044295835790313,1.7414846266734245);
\draw [line width=1.2pt,,color=qqqqff] (8.991963408663445,0.9614459332207967)-- (11.005793806365977,0.5416631896736419);
\draw [line width=1.2pt,,color=qqqqff] (3.006119596912251,-2.006139943894078)-- (4.996380001776475,-2.404206153365795);
\draw [line width=1.2pt,dash pattern=on 2pt off 2pt,color=aqaqaq] (0.004198544585103825,-2.6024579172711846)-- (4.0284626283218925,0.5969560011720245);
\draw [line width=1.2pt,dash pattern=on 2pt off 2pt,color=aqaqaq] (4.0284626283218925,0.5969560011720245)-- (10.044295835790313,1.7414846266734245);
\draw [line width=1.2pt,dash pattern=on 2pt off 2pt,color=aqaqaq] (5.014354222068339,-0.4170332446906455)-- (11.005793806365977,0.5416631896736419);
\draw [line width=1.2pt,dash pattern=on 2pt off 2pt,color=aqaqaq] (1.999480781125392,-3.002589121236023)-- (5.014354222068339,-0.4170332446906455);
\draw [line width=2.pt,,color=qqqqff] (10.011739140043158,0.38443414781475027)-- (11.005793806365977,0.5416631896736419);
\draw [->,line width=1.2pt,] (5.014354222068339,-3.607187767800357) -- (6.150953231922584,-3.8174433525683593);
\draw [->,line width=1.2pt,] (-3.,2.) -- (-3.,3.);
\draw (5.88055870482194,-3.8194892433586103) node[anchor=north west] {$y_1$};
\draw (-3.6936928487376413,3.3136285625850808) node[anchor=north west] {$y_3$};
\draw (-3.6979438450921536,-1.8665134344694385) node[anchor=north west] {$0$};
\draw (-2.513125974929697,-2.102869723308428) node[anchor=north west] {$1$};
\draw (-1.55800446217049,-2.321047867727923) node[anchor=north west] {$2$};
\draw (-0.5434902346047845,-2.5377077991636633) node[anchor=north west] {$3$};
\draw (0.5413276355576718,-2.7740640880026527) node[anchor=north west] {$4$};
\draw (1.5458418631233772,-2.9316349472286455) node[anchor=north west] {$5$};
\draw (2.5306597332858336,-3.1089021638578877) node[anchor=north west] {$6$};
\draw (3.474566675658037,-3.305258452696877) node[anchor=north west] {$7$};
\draw (4.339688188417245,-3.4422220267293686) node[anchor=north west] {$8$};
\draw (-3.557032917301901,-0.7438210624842386) node[anchor=north west] {$1$};
\draw (-3.557032917301901,0.24099680767821752) node[anchor=north west] {$2$};
\draw (-3.5558183469148974,1.284903750050421) node[anchor=north west] {$3$};
\draw (-3.5567292747051497,2.289417977616126) node[anchor=north west] {$4$};
\draw (6.157825921451182,-3.266473023083881) node[anchor=north west] {$1$};
\draw (7.004769289790894,-3.1089021638578877) node[anchor=north west] {$2$};
\draw (8.0092835173566,-2.9119385898253966) node[anchor=north west] {$3$};
\draw (9.033494102325553,-2.675582300986407) node[anchor=north west] {$4$};
\draw (10.097097402101008,-2.4786187269539156) node[anchor=north west] {$5$};
\draw (11.062218914860216,-2.2619587955181756) node[anchor=north west] {$6$};
\begin{scriptsize}
\draw [fill=ududff] (-3.,-2.) circle (0.5pt);
\draw [fill=ududff] (-3.,2.) circle (0.5pt);
\draw [fill=ududff] (5.014354222068339,-3.607187767800357) circle (0.5pt);
\draw [fill=ududff] (3.002828217348916,-0.809909417487408) circle (0.5pt);
\draw [fill=ududff] (5.014354222068339,0.40014919472662075) circle (0.5pt);
\draw [fill=ududff] (11.017502142402869,-2.3971291555863288) circle (0.5pt);
\draw [fill=ududff] (3.0185432642607863,3.165997451215829) circle (0.5pt);
\draw [fill=ududff] (11.001787095490998,1.5473476192931674) circle (0.5pt);
\draw [fill=xdxdff] (-1.981145089079325,-2.2043197872993976) circle (0.5pt);
\draw [fill=xdxdff] (-1.0050955778163269,-2.400055437584214) circle (0.5pt);
\draw [fill=xdxdff] (0.004198544585103825,-2.6024579172711846) circle (0.5pt);
\draw [fill=xdxdff] (0.9892962733578177,-2.800008084870692) circle (0.5pt);
\draw [fill=xdxdff] (1.999480781125392,-3.002589121236023) circle (0.5pt);
\draw [fill=xdxdff] (3.019962370894195,-3.2072351205138796) circle (0.5pt);
\draw [fill=xdxdff] (4.023197222848025,-3.4084224822917415) circle (0.5pt);
\draw [fill=xdxdff] (6.020420031927832,-3.404394397750145) circle (0.5pt);
\draw [fill=xdxdff] (7.011384378493428,-3.204645039934567) circle (0.5pt);
\draw [fill=xdxdff] (8.00104000042622,-3.005159482633978) circle (0.5pt);
\draw [fill=xdxdff] (9.041661473740941,-2.7954007039815343) circle (0.5pt);
\draw [fill=xdxdff] (10.041639259131166,-2.593834501167274) circle (0.5pt);
\draw [fill=xdxdff] (-3.,-1.) circle (0.5pt);
\draw [fill=xdxdff] (-3.,0.) circle (0.5pt);
\draw [fill=xdxdff] (-3.,1.) circle (0.5pt);
\draw [fill=xdxdff] (-1.9944215516257373,-1.8006390657775049) circle (0.5pt);
\draw [fill=xdxdff] (-0.9574827512165562,-1.5950607856191104) circle (0.5pt);
\draw [fill=xdxdff] (8.5465886982794E-4,-1.4050656224530043) circle (0.5pt);
\draw [fill=xdxdff] (1.004554247594701,-1.2060771815109497) circle (0.5pt);
\draw [fill=xdxdff] (2.0016643770515277,-1.008395133678078) circle (0.5pt);
\draw [fill=xdxdff] (3.006803748588471,0.1958999861200208) circle (0.5pt);
\draw [fill=xdxdff] (3.0107166749699954,1.1858703606457541) circle (0.5pt);
\draw [fill=xdxdff] (3.0147543200490867,2.2073945656558065) circle (0.5pt);
\draw [fill=xdxdff] (-1.9770591075282822,2.198178599205096) circle (0.5pt);
\draw [fill=xdxdff] (-0.9695324884308643,2.3933709270356047) circle (0.5pt);
\draw [fill=xdxdff] (0.019913173506579618,2.5850603557361422) circle (0.5pt);
\draw [fill=xdxdff] (0.9854090689059771,2.772109896424841) circle (0.5pt);
\draw [fill=xdxdff] (2.0058499656337423,2.9698041610430597) circle (0.5pt);
\draw [fill=xdxdff] (3.9935735117706646,2.9683043104805584) circle (0.5pt);
\draw [fill=xdxdff] (5.000024835276176,2.7642403610296378) circle (0.5pt);
\draw [fill=xdxdff] (6.005860132946482,2.560301314454753) circle (0.5pt);
\draw [fill=xdxdff] (7.020668939815819,2.3545428358966) circle (0.5pt);
\draw [fill=xdxdff] (8.014470226805303,2.1530437560542435) circle (0.5pt);
\draw [fill=xdxdff] (8.99088995220197,1.9550688904718487) circle (0.5pt);
\draw [fill=xdxdff] (10.044295835790313,1.7414846266734245) circle (0.5pt);
\draw [fill=xdxdff] (11.013620150163108,-1.4227491034062454) circle (0.5pt);
\draw [fill=xdxdff] (11.009487722903831,-0.3855098613277623) circle (0.5pt);
\draw [fill=xdxdff] (11.005793806365977,0.5416631896736419) circle (0.5pt);
\draw [fill=xdxdff] (5.014354222068339,-2.617139812352516) circle (0.5pt);
\draw [fill=xdxdff] (5.014354222068339,-1.6113768099928036) circle (0.5pt);
\draw [fill=xdxdff] (5.014354222068339,-0.4170332446906455) circle (0.5pt);
\draw [fill=xdxdff] (3.9985920013259286,-1.0071096962750126) circle (0.5pt);
\draw [fill=uuuuuu] (5.014354222068338,-1.2082704497945882) circle (0.5pt);
\draw [fill=xdxdff] (6.020046759883498,-1.4074370112050412) circle (0.5pt);
\draw [fill=xdxdff] (7.029779264919688,-1.6074036445553457) circle (0.5pt);
\draw [fill=xdxdff] (8.012410827053339,-1.8020032284288727) circle (0.5pt);
\draw [fill=xdxdff] (9.000460288901632,-1.9976757689125544) circle (0.5pt);
\draw [fill=xdxdff] (10.019748609940669,-2.199534828569305) circle (0.5pt);
\draw [fill=xdxdff] (8.991963408663445,0.9614459332207967) circle (0.5pt);
\draw [fill=uuuuuu] (2.033886234509821,-2.1992696825552556) circle (0.5pt);
\draw [fill=uuuuuu] (4.016823964518636,-2.5997915253164083) circle (0.5pt);
\draw [fill=uuuuuu] (4.996380001776475,-2.404206153365795) circle (0.5pt);
\draw [fill=ududff] (10.011739140043158,0.38443414781475027) circle (0.5pt);
\draw [fill=uuuuuu] (3.006119596912251,-2.006139943894078) circle (0.5pt);
\draw [fill=xdxdff] (4.0284626283218925,0.5969560011720245) circle (0.5pt);
\draw [fill=black] (-1.9982549679031234,1.8000284799159796) circle (0.5pt);
\draw [fill=black] (-0.960524799350616,1.5928734927750765) circle (0.5pt);
\draw [fill=black] (0.023893184605511042,1.3963608529891909) circle (0.5pt);
\draw [fill=black] (1.0332102776042076,1.1948777740951042) circle (0.5pt);
\draw [fill=black] (2.0077055826842116,1.0003459309336955) circle (0.5pt);
\draw [fill=black] (3.009193178687044,0.8004258010591602) circle (0.5pt);
\draw [fill=black] (6.150953231922584,-3.8174433525683593) circle (0.5pt);
\draw [fill=black] (-3.,3.) circle (0.5pt);
\end{scriptsize}
\end{tikzpicture}
\caption{$C(W)$ for the many-to-one market of Example \ref{example:core_workers}} \label{fig:core}
\end{figure}
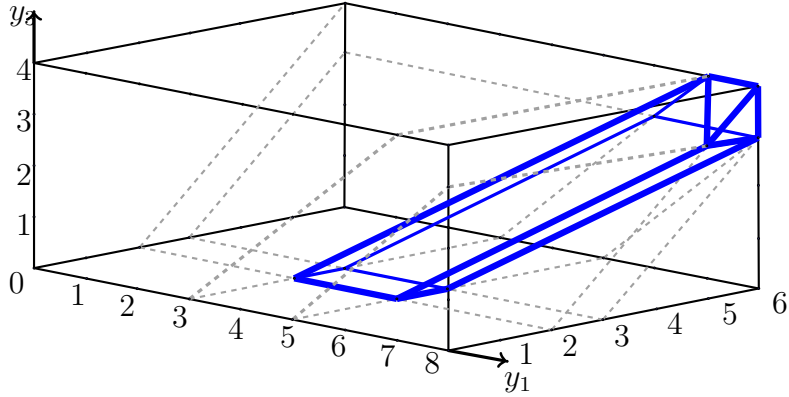
\end{example}

From the description of the core in Proposition \ref{core:description_workers}, it follows straightforwardly that it coincides with the set of competitive equilibrium payoff vectors, that is, $C(W)$ coincides with the set of competitive salary vectors. To see that, we adapt the usual definition of competitive prices to our job market setting with salaries.

Let $\gamma=(F,W,A,r)$ be a many-to-one job market, $\mu\in\mathcal{M}(F,W,r)$ a matching and $y\in\mathbb{R}^W_{+}$ a vector such that $y_j$ is the salary of worker $j\in W$. {\em The pair $(\mu,y)$ is a competitive equilibrium} for this market if and only if:
\begin{enumerate}
\item For each $i\in F$, $\mu(i)\in D_i(y)$, where $D_i(y)$ is the set of $R\subseteq W$, $|R|\le r_i$ such that 
$$\sum_{j\in R}(a_{ij}-y_j)\ge \sum_{j\in S}(a_{ij}-y_j), \mbox{ for all }S\subseteq W\mbox{ such that }|S|\le r_i\; \mbox{ and }$$
\item $y_j=0$ if worker $j$ is not assigned by $\mu$ to any firm $i\in F$.
\end{enumerate}
We then say that $y\in\mathbb{R}^W$ is a vector of competitive salaries and it is compatible with matching $\mu$. It is well known that $\mu$ must be an optimal matching for $\gamma$ and that the competitive salary vector $y$ is  also compatible with any other optimal matching.

On the domain of many-to-one assignment markets, a {\em competitive (equilibrium) rule} $\varphi$ assigns to each market $(F,W,h,t,r)$ a competitive equilibrium $\varphi(h,t)=(\mu, y)$. If we assume $(F,W,r)$ fixed, then a rule assigns a competitive equilibrium to each profile of reported valuations $(h,t)$. Such a rule is {\em non-manipulable} by firm $i\in F$ if reporting its true valuations $h=(h_{ij})_{j\in W}$ is a dominant strategy for firm $i$. Similarly, a rule $\varphi$ is non-manipulable by worker $j\in W$ if reporting his/her true valuation $t_j$ is a dominant strategy.

\section{Core, kernel and bargaining set}
\label{sec:core_geo}

Another relevant result for the one-to-one assignment game is the coincidence between the core and the bargaining set \citep{s99}. The bargaining set is a set-solution concept for coalitional games based on a notion of objections and counterobjections \citep{dm67}. Whenever the core is non-empty, the bargaining set contains the core. The coincidence between the core and the bargaining set, when it holds, is a robustness property of the core since it guarantees that any allocation outside the core has a justified objection (an objection that has no counterobjection) and hence can be dismissed when looking for a cooperative agreement on the distribution of the worth of the grand coalition.

Next, we focus on the relationship between the core and the bargaining set for many-to-one assignment markets. We consider another set-wise solution concept known as the kernel. The \textit{kernel}, introduced by \cite{dm65}, is a non-empty subset of the \textit{classical bargaining set} of \cite{dm67} for all games with a non-empty set of imputations (individually rational and efficient payoff vectors). Whenever the core is non-empty, the intersection between the kernel and the core is non-empty. 


To this end, let us recall the definition of the kernel of a coalitional game. Since our game is zero-monotonic, the kernel is defined as follows:
$$\mathcal{K}(v)=\{z\in I(v)\mid s_{ij}(z)=s_{ji}(z), \mbox{ for all }i,j\in F\cup W\},$$
where $s_{ij}(z)=\max_{i\in S, j\not\in S}\{v(S)-z(S)\}$ is the maximum excess at imputation $z$ of coalitions containing $i$ and not containing $j$. An imputation in the kernel is pairwise balanced in the sense that the maximum excess that agent $i$ can attain with no cooperation of $j$ equals the maximum excess that $j$ can attain without the cooperation of $i$.

Next example shows that the inclusion of the kernel in the core, and also the coincidence of the core with the bargaining set, do not carry over to the many-to-one assignment game. 
We remark that this five-player counter-example is of the smallest size possible, since for any balanced game with at most four players the bargaining set and the core coincide \citep{soly02}, hence the kernel is included in the core \citep{p66}. 
\begin{example}
\label{ex:kernel}
Consider a many-to-one assignment market $\gamma=(F,W,A,r)$ where $F=\{f_{1},f_{2}\}$ is the set of firms, $W=\{w_{1},w_{2},w_{3}\}$ is the set of workers, and the capacities of the firms are $r=(2,2)$. The per-unit pairwise valuation matrix is the following:
 $$A =\bordermatrix{~ & w_{1} & w_{2} & w_{3} \cr
                  f_{1} & 1 & 1  & 1\cr
                  f_{2} & 1 & 1 & 1}.$$
The core contains a unique point, $C(v_{\gamma})=\{(0,0;1,1,1)\}$. One can easily check that imputation $(x,y)=(1,1;1/3,1/3,1/3)$ is not a core allocation, but it lies in the kernel since for any pair of agents the maximum surplus is the same. Let us show that this is indeed the case. The corresponding many-to-one assignment game $(F\cup W,v_{\gamma})$ and the excesses at imputation $(x,y)=(1,1;1/3,1/3,1/3)$, that is, $e(S,(x,y))$, are: $$
\begin{array}{c|cc|ccc|c}
v_{\gamma}(S)  & x_1 & x_2 & y_1 & y_2 & y_3 & e(S,(x,y))\\
\hline
0 & 1 & . & . & . & . &  -1\\ 
0 & . & 1 & . & . & . & -1 \\ 
0 & . & . & 1 & . & . & -1/3 \\ 
0 & . & . & . & 1 & . & -1/3 \\ 
0 & . & . & . & . & 1 & -1/3 \\ 
\hline
0  & 1 & 1 & . & . & . & -2 \\ 
1 & 1 & . & 1 & . & . &  -1/3\\ 
1 & 1 & . & . & 1 & . &  -1/3 \\ 
1 & 1 & . & . & . & 1 &  -1/3 \\ 
1 & . & 1 & 1 & . & . &  -1/3 \\ 
1 & . & 1 & . & 1 & . &  -1/3 \\ 
1 & . & 1 & . & . & 1 &  -1/3\\ 
0  & . & . & 1 & 1 & . &  -2/3\\ 
0  & . & . & 1 & . & 1 &  -2/3 \\ 
0  & . & . & . & 1 & 1 &  -2/3 \\ 
\hline
\end{array} 
\qquad\
\begin{array}{c|cc|ccc|c}
v_{\gamma}(S)  & x_1 & x_2 & y_1 & y_2 & y_3 & e(S,(x,y))\\
\hline
1 & 1 & 1 & 1 & . & . & -4/3 \\ 
1 & 1 & 1 & . & 1 & . & -4/3 \\ 
1 & 1 & 1 & . & . & 1 & -4/3 \\ 
2 & 1 & . & 1 & 1 & . &  1/3 \\ 
2 & 1 & . & 1 & . & 1 & 1/3\\ 
2 & 1 & . & . & 1 & 1 &  1/3\\ 
2 & . & 1 & 1 & 1 & . &  1/3 \\ 
2 & . & 1 & 1 & . & 1 & 1/3 \\ 
2 & . & 1 & . & 1 & 1 &  1/3 \\ 
0 & . & . & 1 & 1 & 1 & -1 \\ 
\hline
2 & 1 & 1 & 1 & 1 & . & -2/3 \\ 
2 & 1 & 1 & 1 & . & 1 & -2/3 \\ 
2 & 1 & 1 & . & 1 & 1 & -2/3\\ 
2 & 1 & . & 1 & 1 & 1 & 0 \\ 
2 & . & 1 & 1 & 1 & 1 & 0 \\ 
\hline
3 & 1 & 1 & 1 & 1 & 1 & 0 \\ 
\hline
\end{array} 
$$ 

Next, we calculate the maximum excess at the point $(x,y)=(1,1;1/3,1/3,1/3)$ for all possible pairs of agents. Since firms (workers) have the same payoff, $x_1=x_2=1$ ($y_1=y_2=y_3=1/3$), it is sufficient to check the maximum excess for the following pairs:
\begin{itemize}
\item If $i=f_1$ and $j=f_2$, then $\max\limits_{\substack{f_{1}\in S\\  f_{2}\notin S}} e(S,(x,y))=1/3=\max\limits_{\substack{f_{2}\in S\\  f_{1}\notin S}} e(S,(x,y))$. 
\item If $i=w_1$ and $j=w_2$, then $\max\limits_{\substack{w_{1}\in S\\  w_{2}\notin S}} e(S,(x,y))=1/3=\max\limits_{\substack{w_{2}\in S\\ w_{1}\notin S}} e(S,(x,y))$. 
\item If $i=f_1$ and $j=w_1$, then $\max\limits_{\substack{f_{1}\in S\\  w_{1}\notin S}} e(S,(x,y))=1/3=\max\limits_{\substack{w_{1}\in S\\  f_{1}\notin S}} e(S,(x,y))$. 
\end{itemize}

Then, for any pair of agents the maximum excess over coalitions containing one of them but not the other is equal at the point $(x,y)=(1,1;1/3,1/3,1/3)$. Hence, the imputation $(x,y)=(1,1;1/3,1/3,1/3)$ lies in the kernel. Since it is not a core allocation, it implies that the kernel is not a subset of the core.

In fact, making use of the facts that (i) the two firms,  as well as any two workers are symmetric players\footnote{Two players  $i$ and $j$ are symmetric in a coalitional game $(N,v)$ if $v(S\cup\{i\})=v(S\cup\{j\})$ for all $S\subseteq N\setminus\{i,j\}$.} in this game, and (ii) symmetric players get the same payoff at any point of the kernel, it can be easily proved that the kernel of this game is the set $\{(\alpha,\alpha;1-\frac{2\alpha}{3}, 1-\frac{2\alpha}{3}, 1-\frac{2\alpha}{3})\mid 0\le \alpha\le 3/2\}$. It is the line segment with the core element $(0,0;1,1,1)$ as one extreme point, and the imputation $(3/2, 3/2; 0,0,0)$ where the firms take all the value as the other extreme point. Notice that the kernel allows for more fair distributions of the value of the market between firms and workers, than the one proposed by the core.\footnote{For example, at the kernel-allocation $(3/4, 3/4; 1/2, 1/2, 1/2)$ the shares of the two sides are equal. In this example, the Shapley value \citep{s53}, is $\phi(v)=( 78/120, 78/120; 68/120, 68/120, 68/120 )$ and it is also in the kernel of the game.}

Since the kernel is a subset of the classical bargaining set, all these points $(\alpha,\alpha;1-\frac{2\alpha}{3}, 1-\frac{2\alpha}{3}, 1-\frac{2\alpha}{3})$ that are in the kernel are also in the bargaining set. This implies that the core and the classical bargaining set do not coincide.
\end{example}

\begin{corollary}
\label{cor:kernel} In the many-to-one assignment game, 
\begin{enumerate}[(i)]
\item The kernel needs not be a subset of the core.
\item The classical bargaining set needs not coincide with the core.
\end{enumerate}
\end{corollary}

Corollary \ref{cor:kernel} implies that the coincidence between the classical bargaining set and the core cannot be carried over from the one-to-one case to the many-to-one and to the many-to-many case. Our result showing that the coincidence between the core and the classical bargaining set is not satisfied is a remarkable exception among several classes of related combinatorial optimization games. \cite{setal03} showed the coincidence result for permutation games. \cite{sol08} proved (among other variants) that the classical bargaining set coincides with the core for several classes including one-to-one assignment games, tree-restricted superadditive games, and simple network games. \cite{b16} obtained the same result for veto games. Recently, \cite{b21} obtained the coincidence result for the so-called quasi-hyperadditive games which contain one-to-one assignment games. \cite{as18} extended the coincidence result from one-to-one assignment games to a class of multi-sided assignment  games known as the supplier-firm-buyer games.

Although we have learned that the kernel may have imputations outside the core, we know that, whenever the core is non-empty, like in the case of the many-to-one assignment games,  the kernel always contains some core elements. The reason is that, for games with non-empty core, it is well-known that the nucleolus is always in the intersection of the kernel and the core.

As in the one-to-one assignment game \citep{gg92} some simplifications can be done to obtain those core allocations that also belong to the kernel. First, only essential coalitions are to be taken into account, and secondly, not all pairs of agents need to be considered.

We have defined a matching as a set of firm-worker pairs that do not violate the capacities of firms and workers, but we can also understand it as a partition of $F\cup W$ in essential coalitions. If $(i, j_1), (i,j_2), \ldots, (i,j_k)$ are pairs in a matching $\mu$, then the coalition $T=\{i,j_1,j_2,\ldots,j_k\}$ is one element of the partition of $F\cup W$ induced by $\mu$. This fact will simplify notations in the next result.

\begin{proposition}
Let $\gamma=(F,W,A,r)$ be a many-to-one assignment game. Then,
$$\mathcal{K}(v_\gamma)\cap C(v_\gamma)=\{z\in C(v_\gamma)\mid s_{ij}(z)=s_{ji}(z)\mbox{ for all }\{i,j\}\in T\subseteq \Phi(A)\},$$
where $\Phi(A)$ is the set of essential coalitions that belong to all optimal matchings.
\end{proposition}
\begin{proof}
Let $z\in C(v_\gamma)$, $S$ an arbitrary coalition of $F\cup W$ and $\mu_S=\{T_1,T_2,\ldots, T_r\}$ an optimal matching for coalition $S$. Then,
$$e(S,z)=v_\gamma(S)-z(S)=\sum_{T\in\mu_S}(v_\gamma(T)-z(T))\le v_\gamma(T_k)-z(T_k), \mbox{ for all }T_k\in\mu_S,$$
where the inequality follows from the fact that excesses at a core allocation are always non-positive. Then, the maximum excess at $z$ over coalitions containing agent $i$ and not containing agent $j$ is always attained at an essential coalition. This implies that only essential coalitions are to be considered to find those core elements that belong to the kernel.

Moreover, if we take two firms $i_1,i_2\in F$, then $$s_{i_1i_2}(z)=e(S,z)=0=e(T,z)=s_{i_2i_1}(z),$$ where $S=\{i_1\}\cup\mu(i_1)$ and $T=\{i_2\}\cup\mu(i_2)$, for any optimal matching $\mu$.

Similarly, if we take two workers that are not assigned to the same firm in an optimal matching $\mu$, that is, $(i_1,j_1)\in\mu$ and $(i_2,j_2)\in\mu$, then
 $$s_{j_1j_2}(z)=e(S,z)=0=e(T,z)=s_{j_2j_1}(z),$$ where $S=\{i_1\}\cup\mu(i_1)$ and $T=\{i_2\}\cup\mu(i_2)$.
 
Finally, if we take a firm $i_1$ and a worker $j_2$ that are not matched in some optimal matching $\mu$, that is, there is $\mu\in\mathcal{M}_{A}(F,W,r)$ and $i_2\in F\setminus\{i_1\}$ such that $(i_2,j_2)\in\mu$, then also $s_{i_1j_2}(z)=e(S,z)=0=e(T,z)=s_{j_2i_1}(z),$ where $S=\{i_1\}\cup\mu(i_1)$ and $T=\{i_2\}\cup\mu(i_2)$.
 
To sum up, only firm-worker pairs that are matched in all optimal matchings and pairs of workers that are matched to the same firm in all optimal matchings are to be considered. 
\end{proof}

As a consequence, if a market $(F,W,A,r)$ is such that there is no essential coalition that belongs to all optimal matchings, then all core elements are in the kernel. This is precisely the case of Example \ref{ex:kernel}. 


\bigskip
To complete this section, we include some remarks on some single-valued solutions: the fair-division point, the nucleolus, and the $\tau$-value.\footnote{The \emph{nucleolus} of a coalitional game $(N,v)$ is the payoff vector $\eta(v)\in\mathbb{R}^N$ that lexicographically minimizes the vector of decreasingly ordered excesses of coalitions among all possible efficient payoff vectors \citep{s69}. The \emph{fair-division point} of a one-to-one assignment market is the midpoint of the buyer-optimal and the seller-optimal core allocations \citep{t81}. The $\tau$-value is a single-valued solution for coalitional games introduced in \citep{tijs81}. It is known that for one-to-one assignment games the $\tau$-value and the fair-division point coincide \citep{nr02}.}

Regarding the computation of the nucleolus of the many-to-one assignment game, we simply point out that only the essential coalitions need to be considered in the lexicographic minimization procedure. This is a notable simplification but still of exponential size.
It remains for further research to see whether the quadratic many essential coalitions which describe the core (Proposition~\ref{core:description_workers}) are also sufficient to determine the nucleolus. 

Secondly, because of 
the existence of an optimal core allocation for each side, which is the worst core allocation for the opposite side, the fair division point can be defined, as in the one-to-one assignment game, as the midpoint of these two core allocations. The convexity of the core guarantees this midpoint is also a core allocation.
For the one-to-one assignment game, the fair division point coincides with the $\tau$-value, which is the only efficient point in the segment between the utopia point, where each agent obtains his/her marginal contribution, and the minimum rights point (see \citep{tijs81} for the formal definition).

Consider the many-to-one assignment market introduced in Example \ref{example:core_workers}. One can calculate, or check with help of the Kohlberg-criterion, that the nucleolus is $(4,2.25$; $5.75,4.25,1.75)$ \footnote{The nucleolus can be calculated efficiently by means of the algorithm introduced by \cite{betal21}.}, the fair division point is $(4.5,2;5.5,4,2)$, and the $\tau$-value is $(143/28,52/28$; $149/28,108/28,52/28)$. The $\tau$-value does not coincide with the fair division point and moreover it does not lie in the core since, for instance, the core constraint for the coalition $\{f_{2},w_{3}\}$ does not hold as $52/28+52/28= 104/28<4=v_{\gamma}(\{f_{2},w_{3}\})$. Thus, we observe that for the many-to-one assignment game, and hence for the transportation game, the $\tau$-value need not be a core allocation. 







\section{Maximum and minimum competitive salaries}
\label{sec:max-min}

In spite of what we observe in Example \ref{ex:kernel} above, generically the core of a many-to-one assignment game contains infinitely many allocations, each of them supported by a vector of competitive salaries. In particular, this is the case when the optimal matching is unique, and then the dimension of the core is $(m+n)-m=n$.

Special attention has been paid to the vectors of maximum and minimum competitive prices (or salaries in our case).  For many-to-one assignment market $\gamma=(F,W,A,r)$, as a particular case of the model in \citep{gs99}, and normalizing at zero the reservation values of the workers, these two extreme vectors of competitive salaries can be obtained from
\begin{eqnarray}\label{maxp}
\overline{y}_j&=&v_{\gamma}(F\cup W)-v_{\gamma}(F\cup (W\setminus\{j\}), \mbox{ for all }j\in W,\\
\label{minp}
\underline{y}_j&=&v_{\gamma^{+j}}(F\cup (W\cup\{j'\}))-v_{\gamma}(F\cup W),\mbox{ for all }j\in W,\end{eqnarray}
where the value $v_{\gamma^{+j}}(F\cup (W\cup\{j'\}))$ is obtained by duplicating in the valuation matrix $A$ the column of worker $j$ and looking for an optimal matching among those that do not allocate the two copies of worker $j$ to the same firm.

It is well known that if all firms have a unitary capacity, then $\underline{y}_j=v_{\gamma}((F\setminus\{j^\mu\})\cup W)-v_{\gamma}((F\setminus\{j^\mu\})\cup (W\setminus\{j\}))$ and the maximum core payoff of the firm assigned to $j$ is its marginal contribution $\overline{x}_{j^\mu}=v_{\gamma}(F\cup W)-v_{\gamma}((F\setminus\{j^\mu\})\cup W)$. In the many-to-one case, those firms with capacity greater than one may not attain their marginal contribution in the core. Take for instance firm $f_1$ in Example \ref{example:core_workers}, since the minimum competitive salaries are $\underline{y}=(3,2,0)$, the maximum core payoff of this firm is $(8-3)+(6-2)=9$ that is below $v_{\gamma}(F\cup W)-v_{\gamma}((F\setminus\{f_1\})\cup W)=11$.

There exist in the literature different axiomatic characterizations of the minimum and/or the maximum competitive equilibrium rules for the one-to-one assignment games, some of them based on non-manipulability properties:\footnote{Other axiomatic characterizations of the optimal competitive rules in the one-to-one case but based on monotonicity properties of the valuations can be found in \citep{vbnr21}.} the firm-optimal (worker-optimal) competitive rules are the only competitive rules that cannot be manipulated by any firm (worker), see for instance \citep{pcs17}. This same reference shows that this characterization can be extended to the worker-optimal competitive rule in the many-to-one assignment market. Moreover, a worker (even with unitary capacity) may manipulate any competitive rule that does not assign this worker the maximum competitive salary.

It is also well known that, in the many-to-one case, the firm-optimal competitive rules (that is, the rules that select the minimum competitive salaries) are manipulable by firms with capacity greater than one, see for instance \citep{so02}. However, the firm-optimal competitive rules cannot be manipulated by a firm $i\in F$ by constantly overreporting its valuations \citep{dn22}, that is, by reporting $h'_{ij}=h_{ij}+c$ for all $j\in W$ and $c>0$. 

We now characterize the minimum competitive rules not by being non-manipulable by workers but by being invariant under the constant decrease of a firm's valuations.

\begin{definition}\label{defINV}
A rule $\varphi=(\mu, y)$ on the domain of many-to-one assignment games is \emph{invariant under the constant decrease of a firm's valuation (INV)} if, for all market $\gamma=(F,W,h,t,r)$, all $i_0\in F$ and $c\ge 0$, if 
\begin{description}
\item{(i)} $h^c_{i_0j}=h_{i_0j}-c$ and $a^c_{i_0j}=\max\{h_{i_0j}-c-t_j,0\}$, for all $j\in W$, 
\item{(ii)} $h^c_{ij}=h_{ij}$, $a_{ij}^c=a_{ij}$ for all $i\in F\setminus\{i_0\}$ and $j\in W$,
\item{(iii)} $c\le a_{i_0j}$ for all $(i_0,j)\in\mu$ and all $\mu\in\mathcal{M}_A(F,W,r)$ and 
\item{(iv)} $\mathcal{M}_A(F,W,r)\subseteq \mathcal{M}_{A^c}(F,W,r)$,
\end{description}
then $y_j(a)=y_j(a^c)$ for all $j\in \mu(i_0)$.
\end{definition}

This invariance states that if a firm decreases all its valuations in a constant amount until a certain threshold $c^*$ such that all initial optimal matchings are still optimal, then the  competitive salaries chosen by the rule for the workers hired by this firm do not change. For instance, if this firm has to pay a fee for each contract it signs, then this fee is entirely paid by the firm and does not affect the salary of the workers that are hired.

\begin{theorem}
The minimum competitive rule is the only competitive rule that is invariant under the constant decrease of a firm's valuations.
\end{theorem}
\begin{proof}
It is proved in \citep{dn22} that, for the multiple partners assignment game (where both firms and workers may have capacities greater than one), and under the assumptions of Definition \ref{defINV}, the payoff of any worker for each partnership under any firm-optimal stable rule is invariant. Since in the many-to-one assignment market, the stable payoffs of workers coincide with their competitive equilibrium payoffs, we get that the minimum competitive salary of the workers hired by a firm $i_0$ that decreases its valuations  by the same amount $c$, is invariant, as long as $c$ satisfies  (iii) and (iv):  $\underline{y}_j(a)=\underline{y}_j(a^c)$ for all $j\in\mu(i_0)$, where $\mu$ is compatible with $y$.

To prove uniqueness, let $\varphi=(\mu,y^{\varphi})$ be a competitive rule that is invariant under the constant decrease of a firm's valuations. 
When applying $\varphi$ to any many-to-one market $(F,W,A,r)$, the competitive salaries $y^{\varphi}(a)$ are also competitive in a related one-to-one assignment market $(\tilde{F}, W, \tilde{A})$ where there are as many copies of each firm $i\in F$ as their capacities $r_i$ indicate \citep{s92}. The invariance of the salaries $y^{\varphi}(a)$ under the constant decrease of a firm's valuations implies the invariance of $y^{\varphi}(\tilde{a})=y^{\varphi}(a)$ under the constant decrease of the valuations of all copies of $i\in F$ and, from Theorem A.4 and Corollary A.5 in \citep{dn22}, this implies $y^{\varphi}(\tilde{a})=\underline{y}^{\varphi}(\tilde{a})$. Then,  \citep{s92} 
implies $y^{\varphi}(a)=\underline{y}^{\varphi}(a)$ and hence $\varphi$ is the minimum competitive rule of the many-to-one assignment market.
\end{proof}

It is shown in \citep{dn22}, for more general many-to-many multiple partners assignment markets,  that the maximum $c$ under the conditions (iii) and (iv) in Definition \ref{defINV} is $c^*=\min_{j\in\mu(i)}\{a_{ij}-\underline{y}_j(a)\}$ or, equivalently, it is the minimum $c$ such that the valuation matrix $a^c$ has an optimal matching with a zero entry.


\section{The set of extreme competitive salary vectors}
\label{extreme_vectors_digraph}
Besides the vectors of maximum and minimum competitive salaries, there may be several other extreme points in the set of competitive salaries, which correspond to the set of extreme core allocations. The description of these extreme points gives information about how large this set is, and how many different stable agreements can be attained in the market. 
The digraph we introduce next, associated with each vector of competitive salaries, provides a 
characterization of all extreme vectors of competitive salaries, not just of the maximum and minimum ones.

\begin{definition}
\label{def:tight_digraph}
Let $\gamma=(F,W,A,r)$ be a capacity-balanced many-to-one assignment market and $\mu\in\mathcal{M}_A(F,W,r)$ be an optimal matching. 
For each vector of competitive salaries $y\in C(W)$, we define the \emph{tight digraph} $(W_0,E^{y})$ with set of nodes $W_0:=W\cup\{0\}$, where $0$ is a fictitious worker whose salary is fixed to $y_0=0$ and with the set of arcs $E^{y}$ such that 
$$
(j,k)\in E^{y} \quad \leftrightarrow \quad y_k-y_j = 
		\begin{cases}
		0 & \text{if } j=0, k\in W ;\\
		-a_{j^\mu j} & \text{if } j\in W, k=0 ;\\
        a_{j^\mu k}-a_{j^\mu j} & \text{if } j\in W, k\in W\setminus \mu(j^\mu). \\
		\end{cases}
$$ 
\end{definition}

This tight digraph is inspired by the one introduced in \citep{bg90} and also used in \citep{hetal02} to study extreme core points of the one-to-one assignment game. There, the nodes of the graph consist of the agents on both sides of the market, not just from one side as we do for the many-to-one case. And also, their graph is not directed since it is based on the constraints $x_i+y_j\ge a_{ij}$ where both variables have the same sign. They find that the extreme core points are those core points with a tight graph that has an agent with zero payoff in each component.
In our setting, we replace that property with connectedness to the fixed 0 payoff node, hence connectedness of the underlying undirected graph (the base-graph) of the tight digraph. Besides, we also characterize the maximum and the minimum competitive salary vectors with an (easily verifiable) additional property of the tight digraph.    

Before the general discussion, we illustrate the idea and foreshadow the results on the market situation of Example \ref{example:core_workers}.
\begin{example}
\label{ex:extended_digraph}
We revisit Example \ref{example:core_workers} and introduce a fictitious worker $0$ who is optimally matched to a fictitious firm with capacity 1, denoted $0^\mu$, because in any matching of the extended capacity-balanced market the two fictitious agents are required to be paired. Their payoffs are fixed to 0. The virtual possibility of being matched to the fictitious agent on the other side will represent the outside option of an agent, thus the pairwise surplusses with them are set to zero. The extended market, with the optimally assigned firm-worker pairs boxed, is given on the left below. For brevity, we represent the workers and the firms by their index.  On the right below we present the description of $C(W)$ where all constraints are written in a unified way.
\medskip
$$
\setlength{\extrarowheight}{2pt}
\begin{array}{@{\;}c||@{\,}c@{\,}|@{\,}c@{\,}c@{\,}c@{\,}||@{\,}l}
  & 0 & 1 & 2 & 3 & \\
\hline
\hline
0^\mu & \fbox{0} & 0 & 0 & 0 & r_{0^\mu}=1 \\
\hline
1 & 0 & \fbox{8} & \fbox{6} & 3 & r_{1}=2 \\
2 & 0 & 7 & 6 & \fbox{4} & r_{2}=1 \\
\hline
\end{array}
\qquad\quad
\begin{array}{c@{\;}c@{\;}c@{\;}c@{\;}c@{\;}r}
\hline
-0 & +y_1 &     &     & \geq & 0=0-0 \\
-0 &     & +y_2 &     & \geq & 0=0-0 \\
-0 &     &     & +y_3 & \geq & 0=0-0 \\
\hline
+0 & -y_1 &     &     & \geq & -8=0-8 \\
+0 &     & -y_2 &     & \geq & -6=0-6 \\
+0 &     &     & -y_3 & \geq & -4=0-4 \\
\hline
 & -y_1 &     & +y_3 & \geq & -5=3-8 \\
 &     & -y_2 & +y_3 & \geq & -3=3-6 \\
\hline 
 & +y_1 &     & -y_3 & \geq & 3=7-4 \\
 &     & +y_2 & -y_3 & \geq & 2=6-4 \\
\hline
\end{array}
$$
\medskip

Due to this special structure, we associate a directed graph that represents by arcs the inequalities which are tight (satisfied as equality) at a given $y\in C(W)$  and decide  if $y$ is an extreme point by checking whether the base-graph is connected.

Recall that in this market the minimum competitive salary vector is $(3,2,0)$ that makes the following inequalities tight: $-0+y_3=0$, $y_1-y_3=3$, and $y_2-y_3=2$. The associated digraph is pictured on the left in Figure \ref{fig:digraphs-min-max}.
At the maximum competitive salary vector $(8,6,4)$ the following inequalities are tight: $0-y_1=-8$, $0-y_2=-6$, $0-y_3=-4$, and $y_2-y_3=2$. The associated digraph is pictured on the right in Figure \ref{fig:digraphs-min-max}.

\begin{figure}[h]
\centering
\begin{tikzpicture}[scale=1] 
\path 
	(0,-3.6) node(am0) [rectangle,draw] {$0$}
    (1,-1) node(a11) [rectangle,draw] {1}
    (2.2,-2) node(a22) [rectangle,draw] {2}
    (3.4,-3) node(a33) [rectangle,draw] {3}
 ; 
\draw[-latex] (am0) -- (a33) ;
\draw[-latex] (a33) .. controls +(-0.1,+1.9) .. (a11) ;
\draw[-latex] (a33) .. controls +(-0.2,+0.8) .. (a22) ;
\end{tikzpicture}
\qquad\qquad\qquad\qquad
\begin{tikzpicture}[scale=1] 
\path 
	(0,-3.6) node(a00) [rectangle,draw] {$0$}
    (1,-1) node(a11) [rectangle,draw] {1}
    (2.2,-2) node(a22) [rectangle,draw] {2}
    (3.4,-3) node(a33) [rectangle,draw] {3}
 ; 
\draw[-latex] (a11) -- (a00) ;
\draw[-latex] (a22) -- (a00) ;
\draw[-latex] (a33) -- (a00) ;
\draw[-latex] (a33) .. controls +(-0.2,+0.8) .. (a22) ;
\end{tikzpicture}
\caption{digraph of minimum vector $(3,2,0)$,  
digraph of maximum vector $(8,6,4)$} \label{fig:digraphs-min-max}
\end{figure}
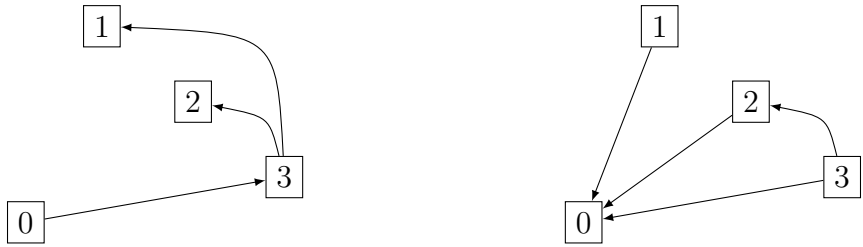

In both cases the base-graph is connected. Notice that in the tight digraph of the minimum competitive salary vector $(3,2,0)$ node 0 is the only source, while in the tight digraph of the maximum competitive salary vector $(8,6,4)$ node 0 is the only sink.

Similarly, the tight digraphs associated with CE vectors $(3,3,0)$ and $(7,6,4)$, pictured, respectively, on the left and on the right in Figure \ref{fig:digraphs-notmin-notmax}, are both connected. Thus, both vectors are also extreme points of $C(W)$. 
However, in neither of these tight digraphs node 0 is the only source or the only sink.

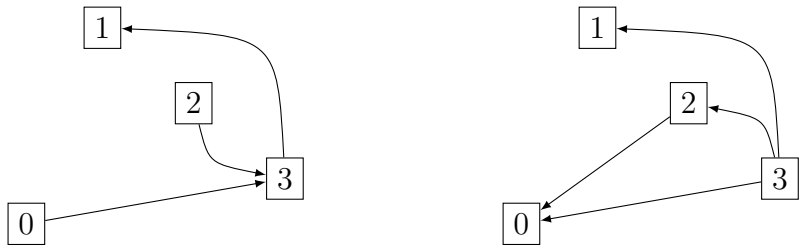
\begin{figure}[h]
\begin{center}
\begin{tikzpicture}[scale=1] 
\path 
	(0,-3.6) node(a00) [rectangle,draw] {$0$}
    (1,-1) node(a11) [rectangle,draw] {1}
    (2.2,-2) node(a22) [rectangle,draw] {2}
    (3.4,-3) node(a33) [rectangle,draw] {3}
 ; 
\draw[-latex] (a00) -- (a33) ;
\draw[-latex] (a22) .. controls +(0.2,-0.8) .. (a33) ;
\draw[-latex] (a33) .. controls +(-0.1,+1.9) .. (a11) ;
\end{tikzpicture}
\qquad\qquad\qquad
\begin{tikzpicture}[scale=1] 
\path 
	(0,-3.6) node(a00) [rectangle,draw] {$0$}
    (1,-1) node(a11) [rectangle,draw] {1}
    (2.2,-2) node(a22) [rectangle,draw] {2}
    (3.4,-3) node(a33) [rectangle,draw] {3}
 ; 
\draw[-latex] (a22) -- (a00) ;
\draw[-latex] (a33) -- (a00) ;
\draw[-latex] (a33) .. controls +(-0.1,+1.9) .. (a11) ;
\draw[-latex] (a33) .. controls +(-0.2,+0.8) .. (a22) ;
\end{tikzpicture}
\end{center}
\caption{digraph of extreme vector $(3,3,0)$
\qquad
digraph of extreme vector $(7,6,4)$} \label{fig:digraphs-notmin-notmax}
\end{figure}

Indeed, in case of $(3,3,0)$, node 2 is also a source, indicating that none of the constraints which contains $-y_2$ is tight, hence $y_2$ can be decreased (with a sufficiently small positive amount) without leaving the feasible solution set. Therefore, $(3,3,0)$ cannot be the minimum competitive salary vector.
Similarly, in case of $(7,6,4)$, node 1 is also a sink, indicating that none of the constraints which contain $+y_1$ is tight, hence $y_1$ can be increased (with a sufficiently small positive amount) without violating any of the lower-bound constraints. Therefore, $(7,6,4)$ cannot be the maximum competitive salary vector. 
\end{example}

We now formally establish the characterization of the extreme vectors of competitive salaries of the many-to-one assignment game by properties of the corresponding tight digraphs (Definition~\ref{def:tight_digraph}).
\begin{theorem}
\label{thm:ext_core_char}
Let $\gamma=(F,W,A,r)$ be a capacity-balanced many-to-one assignment market in which $\mu$ is an optimal matching. Then 
\begin{description}
\item[(A)] $y\in C(W)$ is an extreme vector of competitive salaries if and only if the base-graph of the associated tight digraph $(W_0,E^y)$ is connected. 
\item[(B)]  $y\in C(W)$ is the minimum vector of competitive salaries if and only if its tight digraph $(W_0,E^y)$ contains a 0-sourced directed spanning tree (i.e. all arcs of the spanning tree are directed away from node 0). 
\item[(C)] $y\in C(W)$ is the maximum vector of competitive salaries if and only if its tight digraph $(W_0,E^y)$ contains a 0-sinked directed spanning tree (i.e. all arcs of the spanning tree are directed towards node 0).
\end{description}
\end{theorem}
\begin{proof}
First, we prove the ``\emph{only if}'' part of characterization (A). Suppose on the contrary that the base-graph of the tight digraph associated with an extreme vector $y\in C(W)$ is not connected. Then let $W'\subset W$ be the node set of a component which does not contain node 0. Now, let us define $\varepsilon=\min\{(y_k-y_j)-(a_{j^\mu k}-a_{j^\mu j}) : j\in W', k\in W_0\setminus W'\}$, with $y_k=a_{j^\mu k}=0$ if $k=0$. Since there are no arcs between $W'$ and the rest of the nodes $W_0\setminus W'$, we have $\varepsilon>0$. If we define $y'_{j}=y_{j}+\varepsilon$, $y''_{j}=y_{j}-\varepsilon$ for all $j\in W'$, and $y'_{k}=y''_{k}=y_{k}$ for all $k\in W_0\setminus W'$, both vectors $y'$ and $y''$ also belong to $C(W)$. However, $y=\frac{1}{2}y'+\frac{1}{2}y''\in C(W)$, which contradicts the assumption that $y$ is an extreme point. 

Second, we prove the ``\emph{if}'' part of (A). If the base-graph of the tight digraph associated with a vector $y\in C(W)$ is connected, then there is a path from node 0 to any node $j\in W$, that is a sequence of nodes $0=j_0,j_1,\ldots, j_k=j$ with $k\geq 1$ such that any two consecutive nodes are the two endpoints of an arc. If $(j_{h}j_{h+1})\in E^\mu$ $(0\leq h\leq k-1)$ then it is called a forward arc, if $(j_{h+1}j_{h})\in E^\mu$ $(0\leq h\leq k-1)$ then it is called a backward arc. In case nodes $j_h$ and $j_{h+1}$ $(0\leq h\leq k-1)$ are connected by both types of arcs in the tight digraph, we choose one of them arbitrarily for the path. If we add the equations $y_{h+1}-y_{h}=a_{h^\mu h+1}-a_{h^\mu h}$ related to the forward arcs in this path, and subtract the sum of the equations $y_{h}-y_{h+1}=a_{(h+1)^\mu h}-a_{(h+1)^\mu h+1}$ related to the backward arcs, all variables $y_h$, $1\leq h\leq k-1$ (if any) cancel out, only $y_j-y_0=y_j$ remains on the left side. Thus, we get $y_j=\sum_{(h,h+1)\in E^\mu} (a_{h^\mu h+1}-a_{h^\mu h}) - \sum_{(h+1,h)\in E^\mu} (a_{(h+1)^\mu h}-a_{(h+1)^\mu h+1})$. Since all salaries $y_j$ $(j\in W)$ are uniquely determined by the tight constraints, their vector $y$ is an extreme point of $C(W)$.

We only prove the characterization for the minimum vector of competitive salaries in (B), the proof for the maximum vector in (C) goes in an analogous way.

Assume first that $y\in C(W)$ is the minimum vector of competitive salaries. If there is a node $j\in W$ with no incoming arc then all constraints in which $y_j$ appears with $+1$ coefficient are satisfied as strict inequalities, so while keeping all other variables fixed, we can decrease $y_j$ with a sufficiently small positive amount without violating any of these constraints. Besides, we actually increase the left hand side of those greater-or-equal inequalities in which $y_j$ appears with $-1$ coefficient, and we do not change the left hand side of the rest of the constraints. This would contradict the minimality of $y_j$. Thus, none of the nodes $j\in W$ can be a source at the minimum vector of competitive salaries. A similar argument shows that there must be an arc going out from node 0, for otherwise we could decrease all salaries $y_j$ $(j\in W)$ with a sufficiently small positive amount without violating any CE constraint, again a contradiction to the minimality of vector $y$. Combining these two observations with the finiteness of the number of nodes, we conclude that there must exist a directed path (i.e. containing only forward arcs) from node 0 to any node $j\in W$, implying that the tight digraph $(W_0,E^y)$ must contain a 0-sourced directed spanning tree.

To see the converse implication in (B) for the minimum vector of competitive salaries, assume that for an arbitrary $y\in C(W)$ the associated tight digraph $(W_0,E^y)$ contains a 0-sourced directed spanning tree. Then there exits a directed path $0=j_0,j_1,\ldots, j_k=j$ with $k\geq 1$ from node 0 to any node $j\in W$ containing only forward arcs. Thus, if we add the related inequalities $y'_{h+1}-y'_{h}\geq a_{h^\mu h+1}-a_{h^\mu h}$ along this directed path, we get $y'_j \geq \sum_{(h,h+1)\in E^\mu} (a_{h^\mu h+1}-a_{h^\mu h})$ for any feasible vector $y'\in C(W)$. For the selected $y\in C(W)$, all these constraints hold as equalities, thus, $y_j=\min\{ y'_j: y'\in C(W)\}$ for all $j\in W$, implying that $y\in C(W)$ is the minimum vector of competitive salaries.
\end{proof}

We remark that a tight digraph might contain a directed cycle, which might even contain node 0, but only if there are alternative optimal matchings in the many-to-one assignment market. If the optimal matching is unique, like in Example~\ref{ex:extended_digraph}, node 0 is either a source or a sink (but not both) in the tight digraph associated with any extreme salary vector.

We conclude this section with an immediate consequence of Theorem~\ref{thm:ext_core_char}. 
\begin{corollary}
\label{corollary:FOCAWOCA}
Let $\gamma=(F,W,A,r)$ be a capacity-balanced many-to-one assignment market in which $\mu$ is an optimal matching, and let $y\in C(W)$. If $y$ is the vector of minimum (resp. maximum) competitive salaries, then there is a worker $j\in W$ with salary $y_{j}=0$ (resp. $y_{j}=a_{j^\mu j}$).
\end{corollary}

\section{The max-min salary vectors}
\label{sec:extreme_char}

In this section we intend to compute the set of extreme core allocations or, equivalently, extreme competitive salary vectors of the many-to-one assignment markets. A natural first approach is to consider the relationship between the extreme core allocations and some  lexicographic allocation procedures. This approach has been applied by (i) \cite{hetal02} to show that each extreme core allocation of a one-to-one assignment game is a marginal payoff vector and by (ii) \cite{ietal07} to see that each such extreme point is the result of a  lexicographic minimization procedure on the set of rational allocations: for each order on the set of agents, let the payoff to the first player in the order be zero and, for each following agent, compute the minimum payoff that satisfies all core inequalities with his/her predecessors while preserving the payoffs that they have already been allocated.
More recently,  (iii), \cite{ns17} prove that each extreme core allocation of the one-to-one assignment game is the result of a lexicographic maximization over the set of dual rational allocations (lemarals).  However, it is easy to find examples (see Example \ref{ex:not_INTO_ONTO} in the Appendix) showing that none of these three procedures allows to describe all the extreme core allocations of many-to-one assignment markets.  


The characterization of the extreme competitive salary vectors of the many-to-one assignment game by means of the tight digraphs given in Theorem \ref{thm:ext_core_char} will allow to describe a procedure to obtain all these extreme points. We will see that the extreme competitive salary vectors of these games also correspond to a sequence of  lexicographic optimization, 
where, for each given order, some workers maximize their salary while some other workers minimize it, always preserving what  has been allocated to their predecessors.



There are two main differences between the following definition of the max-min salary vectors and the lexicographic procedures applied to the one-to-one assignment game: only workers are now considered and each order on the set of workers must be completed with an indication of whether the worker in this position maximizes or minimizes his/her salary.

Let $\theta:\{1,\ldots,n=|W|\}\longrightarrow W$ be an order on the set of workers, where $\theta(i)$ is the worker in the $i$th-position, and we can also write $\theta=(j_1,j_2,\ldots, j_n)$. We denote by $\Sigma_W$ the set of all orders on $W$. Given a worker $j\in W$, $P_j^{\theta}=\{k\in W\mid \theta^{-1}(k)<\theta^{-1}(j)\}$ is the set of predecessors of $j$ according the order $\theta$.

Then, an extension of the order $\theta$ is 
$$\begin{array}{rcl}\tilde{\theta}:\{1,\ldots,n=|W|\} &\longrightarrow & W\times\{\min,\max\}\\
                                         i & \mapsto & \tilde{\theta}(i)=\left\{\begin{array}{l}
                                               (\theta(i), \min)=\underline{\theta}(i)\\  \mbox{ or }\\
                                               (\theta(i),\max)=\overline{\theta}(i),\end{array}\right.
\end{array}$$
where $\underline{\theta}(i)$ means that worker is in $i$th position and will minimize his/her salary under some constraints. Similarly, $\overline{\theta}(i)$ means that the $i$th player in the order will maximize his/her salary under some constraints. We denote by $\widetilde{\Sigma}_W$ the set of all extended orders on $W$.
Clearly, $|\Sigma_W|=n!$ and $|\widetilde{\Sigma}_W|=n!\cdot 2^n$, where $n=|W|$ is the number of workers.

\begin{definition}\label{def:maxmin}
Let $(F,W,A,r)$ be a capacity-balanced many-to-one assignment game, $\mu$ an optimal matching, $\theta=(j_1,j_2,\ldots, j_n)$ an order on $W$ and $\tilde{\theta}$ an extension of $\theta$. The related {\em max-min salary vector} $y^{\tilde{\theta}}$ satisfies
$$y_{j_1}^{\tilde{\theta}}=\left\{\begin{array}{lll}
0 & \mbox{ if } & \tilde{\theta}(1)=\underline{\theta}(1)\\
a_{j_1^\mu j_1} & \mbox{ if } & \tilde{\theta}(1)=\overline{\theta}(1),\end{array}\right.$$
and for all $1<r\le n$, 
$$y^{\tilde{\theta}}_{j_r}=\left\{\begin{array}{lll}
\max_{j\in P^{\theta}_{j_r}, j^\mu\ne j_r^\mu }\{y_j-a_{j^\mu j}+a_{j^\mu j_r},0\} & \mbox{ if } & \tilde{\theta}(r)=\underline{\theta}(r)\\
\min_{j\in P^{\theta}_{j_r}, j^\mu\ne j_r^\mu }\{y_j-a_{j_r^\mu j}+a_{j_r^\mu j_r}, a_{j_r^\mu j_r}\} & \mbox{ if } & \tilde{\theta}(r)=\overline{\theta}(r).\end{array}\right.$$
\end{definition}

To give an interpretation to these vectors, recall from Proposition \ref{core:description_workers} that the core constraints worker $j_r$ must satisfy are $0\le y_{j_r}\le a_{j_r^\mu j_r}$  and  
$$a_{j^\mu j_r}-a_{j^\mu j}  \le y_{j_r}-y_j\le a_{j_r^\mu j_r}-a_{j_r^\mu j}, \mbox{ for all }j\in W, \, j^\mu\ne j_r^\mu .$$
Then, when we reach worker $j_r$ following order $\theta$, the max-min vector procedure only considers the core constraints with variables from $P^{\theta}_{j_r}\cup\{j_r\}$ and determines a payoff (salary) for $j_r$ that satisfies (in a tight way) either one lower core bound or one upper core bound, depending on whether the extended order $\tilde{\theta}$ determines $j_r$ is a maximizer or a minimizer. Since all $y_j$ values for $j\in P^{\theta}_{j_r}$ have already been set, finding $y_{j_r}$ amounts to the elementary optimization problems given in the above definition. It is not surprising that a max-min salary vector may not be in $C(W)$, since one half of the core constraints are not checked during the procedure that builds such vector. However, we show next that if a max-min salary vector is competitive, then it is an extreme competitive vector. This same property (the fact that when they are in the core, they are extreme core points) is satisfied by the marginal worth vectors in arbitrary coalitional games and by the max-payoffs vectors in one-to-one assignment games, which are also collections of vectors that are defined for each possible order on a player set.

\begin{proposition}
Let $\gamma=(F,W,A,r)$ be a capacity-balanced many-to-one assignment market, $\mu$ an optimal matching, $\theta$ an order on $W$ and $\tilde{\theta}$ an extension of $\theta$. If  $y^{\tilde{\theta}}\in C(W)$, then $y^{\tilde{\theta}}\in Ext(C(W))$.
\end{proposition}
\begin{proof}
Let $y^{\tilde{\theta}}\in C(W)$. By definition of the max-min vectors, at each step of the procedure one core constraint is tight at $y^{\tilde{\theta}}$. Moreover, these equations are linearly independent since each of them involves a new worker whose salary does not take part in the previous equations. Since the membership in $C(W)$ is guaranteed by the assumption, the fact that $n$ linearly independent constraints are tight at $y^{\tilde{\theta}}$ implies that this is an extreme point of $C(W)$.
\end{proof}

Now the question is whether all extreme points of $C(W)$ in a many-to-one assignment market are of this type, that is, all are max-min salary vectors related to some extended order on the set of workers. Let us consider again the market of Example \ref{example:core_workers}.

\begin{example}\label{first_ex_revisited}
Consider again the many-to-one assignment market $\gamma=(F,W,A,r)$ with set of firms $F=\{f_{1},f_{2}\}$ with capacities $r=(2,1)$, set of workers $W=\{w_{1},w_{2},w_{3}\}$ with unitary capacity and pairwise valuation matrix 
$$A =\bordermatrix{~ & w_{1} & w_{2} & w_{3} \cr
                  f_{1} & 8 & 6 & 3\cr
                  f_{2} & 7 & 6 & 4}.$$
We can obtain the extreme core points from the picture of the salary-core in Figure \ref{fig:core} and then check that all core vertices are supported by max-min salary vectors. 
Another approach is to compute for each of the $3!\cdot 2^3=48$ extended orders the associated max-min salary vectors and check their core membership. The result of this tedious but computationally straightforward excercise is given in Appendix \ref{ex:all_extended_payoffs}. It shows that all 9 extreme core vectors in this market are supported by max-min salary vectors, they are obtained from 28 extended orders, while the remaining 20 extended orders determine max-min salary vectors outside the core.

In the following table we indicate for each core vertex all extended orders such that the related max-min salary vector supports  that vertex.
\renewcommand{\arraystretch}{1.2}
$$
\begin{array}{|c@{\;}c|c@{\;}c@{\;}c||l}
x_1 & x_2 & y_1 & y_2 & y_3 & \textrm{extended order } \\
\hline 
\hline 
0 & 0 & 8 & 6 & 4 & (\overline{1}, \overline{2}, \overline{3})  \textrm{ in any permutation, } (\overline{1},  \overline{3}, \underline{2}), (\overline{3},  \overline{1}, \underline{2}), (\overline{3}, \underline{2}, \overline{1})   \\
0 & 1 & 8 & 6 & 3 & (\overline{1}, \overline{2}, \underline{3}),(\overline{2}, \overline{1}, \underline{3}), (\overline{1}, \underline{3}, \overline{2}), (\overline{2}, \underline{3}, \overline{1})     \\
1 & 1 & 8 & 5 & 3 & (\overline{1}, \underline{3}, \underline{2})   \\
\hline 
1 & 0 & 7 & 6 & 4 & (\overline{2}, \overline{3}, \underline{1}), (\overline{3}, \overline{2}, \underline{1}), (\overline{3}, \underline{1}, \overline{2}), (\overline{3}, \underline{1}, \underline{2}), (\overline{3}, \underline{2}, \underline{1})  \\
2 & 1 & 6 & 6 & 3 & (\overline{2}, \underline{3}, \underline{1}) \\
\hline 
\hline 
9 & 4 & 3 & 2 & 0 & (\underline{3}, \underline{2}, \underline{1}),   (\underline{3}, \underline{1}, \underline{2})\\ 
7 & 4 & 5 & 2 & 0 & (\underline{3}, \underline{2}, \overline{1}), (\underline{3}, \overline{1}, \underline{2})  \\
6 & 4 & 5 & 3 & 0 & (\underline{3}, \overline{2}, \overline{1}), (\underline{3}, \overline{1}, \overline{2})    \\
8 & 4 & 3 & 3 & 0 & (\underline{3}, \overline{2}, \underline{1}), (\underline{3}, \underline{1}, \overline{2})   \\
\hline 
\hline 
\end{array}
$$
\renewcommand{\arraystretch}{1}
Characterization (A) in Theorem \ref{thm:ext_core_char} offers explanations not just for the various multiplicity a given extreme core vector appears as max-min salary vector, but, more importantly, why the full enumeration of min-max salary vectors will always provide all extreme core vectors (see Theorem \ref{thm:extremes are maxmin} below). 

Take for instance extreme competitive salary vector $(7,6,4)$ and consider its tight graph drawn in Example \ref{ex:extended_digraph}. It has 4 arcs on 4 nodes, so its connected base-graph admits multiple spanning trees. We see that any of these spanning trees contains arc $(3,1)$ and two of the other three arcs (which form a cycle). 
For instance, the spanning tree with arcs $(2,0)$, $(3,2)$, $(3,1)$ allows only the order $(2,3,1)$ and make arcs $(2,0)$ and $(3,2)$ backward arcs and arc $(3,1)$ a forward arc. If, starting from node 0, we reach a node with a backward (resp. forward) arc, we set to maximize (resp. minimize) the payoff for that worker. Thus, in this case we get the extended order $(\overline{2}, \overline{3}, \underline{1})$. The related max-min salary vector  is computed as follows:
$$\begin{array}{rcl}
y_2&=&a_{12}=6,\\
y_3&=&\min\{y_2-a_{22}+a_{23}, a_{23}\}=\min\{4,4\}=4,\\
y_1&=&\max\{y_3-a_{23}+a_{21}, 0\}=\max\{7, 0\}=7.
\end{array}$$
On the other hand, the spanning tree with backward arc $(3,0)$ and forward arcs $(3,2)$, $(3,1)$ allows two extended orders compatible with the partial order induced by this 0-rooted spanning tree, namely $(\overline{3}, \underline{2}, \underline{1})$ and $(\overline{3}, \underline{1}, \underline{2})$. 

This example also shows that not all max-min salary vectors belong to the core, and hence they may not lead to an extreme core allocation. Take for instance the extended order $\tilde{\theta}=(\underline{1}, \overline{2}, \underline{3})$. Then, 
$$\begin{array}{rcl}
y_1&=&0,\\
y_2&=&a_{12}=6,\\
y_3&=&\max\{y_1-a_{11}+a_{13}, y_2-a_{12}+a_{13}, 0\}=\max\{-5, 3, 0\}=3.
\end{array}$$
The related max-min salary vector is $y^{\tilde{\theta}}=(0, 6, 3)$ and it does not lead to a core payoff since the constraint $y_3-y_1\le a_{23}-a_{21}$, which were ignored when $y_3$ were minimized, is not satisfied.
\end{example}

Next theorem shows that although a max-min salary vector may not be an extreme core allocation, the converse inclusion always holds. As the example above illustrates, every extreme core point is supported by a max-min salary vector related to one extended order, or maybe to several of them.

\begin{theorem}\label{thm:extremes are maxmin}
Let $\gamma=(F,W,A,r)$ be a capacity-balanced many-to-one assignment market and $\mu$ an optimal matching. Then, 
$$Ext(C(W))\subseteq \{y^{\tilde{\theta}}\}_{\tilde{\theta}\in \widetilde{\Sigma}_W}.$$
\end{theorem}
\begin{proof}
Let $y\in Ext(C(W))$ and consider the related tight digraph $(W_0,E^y)$. From Theorem \ref{thm:ext_core_char} (A), the base-graph is connected, hence there exists $j_1\in W$ such that at least one of $(0,j_1)\in E^y$, meaning $y_{j_1}=0$, or $(j_1,0)\in E^y$, meaning $y_{j_1}=a_{j_1^\mu j_1}$, holds. If both relations hold, we pick one of them. In the first case define $\tilde{\theta}(1)=\underline{\theta}(1)=j_1$ and in the second case $\tilde{\theta}(1)=\overline{\theta}(1)=j_1$. Notice that in both cases $y^{\tilde{\theta}}_{j_1}=y_{j_1}$.

For $1<r\le n-1$, assume by induction hypothesis that there exists $\tilde{\theta}\in \widetilde{\Sigma}_W$ with $\tilde{\theta}(k)=j_k$ for all $1\le k\le r$ such that $y^{\tilde{\theta}}_{\theta(k)}=y^{\tilde{\theta}}_{j_k}=y_{j_k}$, and show this also holds for $r+1$.
\smallskip

\noindent 
{\em Case 1}: There exists some $j\in W\setminus\{j_1,j_2,\ldots, j_r\}$ and some $j_k\in\{j_1,j_2,\ldots, j_r\}$ such that $(j_k, j)\in E^{y}$.

In this case, $y_j-y_{j_k}=a_{j_k^\mu j} - a_{j_k^\mu j_k}$, which implies $y_j=y_{j_k}+a_{j_k^\mu j}-a_{j_k^\mu j_k}$. Then, set $\tilde{\theta}(r+1)=\underline{\theta}(r+1)=j$, that is, $j_{r+1}=j$, and notice that, since $y$ is a vector of competitive salaries, the inequalities $y_j\ge 0$ and $y_j\ge y_{j_h}+a_{j_h^\mu j}-a_{j_h^\mu j_h}$ hold for all $j_h\in\{j_1,\ldots,j_r\}$ with $j_h^\mu\ne j^\mu$. This guarantees that $y_{j_{r+1}}=y^{\tilde{\theta}}_{j_{r+1}}$.
\medskip

\noindent 
{\em Case 2}: There exists some $j\in W\setminus\{j_1,j_2,\ldots, j_r\}$ and some $j_k\in\{j_1,j_2,\ldots, j_r\}$ such that $(j, j_k)\in E^{y}$.

In this case, $y_j-y_{j_k}=a_{j^\mu j}-a_{j^\mu j_k}$, which implies $y_j=y_{j_k}+a_{j^\mu j}-a_{j^\mu j_k}$. Then, set $\tilde{\theta}(r+1)=\overline{\theta}(r+1)=j$, that is, $j_{r+1}=j$, and notice that, since $y$ is a vector of competitive salaries , the inequalities $y_j\le a_{j^\mu j}$ and $y_j\le y_{j_h}+a_{j^\mu j}-a_{j^\mu j_h}$, for all $j_h\in\{j_1,\ldots,j_r\}$ with $j_h^\mu\ne j^\mu$, hold. This shows that $y_{j_{r+1}}=y^{\tilde{\theta}}_{j_{r+1}}$.
\medskip

By connectedness of the base-graph at least one of the above two cases holds, if both hold, we pick one of them, and continue building the spanning tree till all nodes in $W$ are reached. An extended order is constructed such that the associated min-max vector coincides with the extreme core vector $y$, and our inductive proof ends.
\end{proof}





A consequence of the above theorem is that each extreme core point of a many-to-one assignment game is the result of a (computationally very simple) lexicographic optimization procedure carried out by the workers over the core.  This somehow resembles the one-to-one assignment game, where each extreme core point can be obtained from a lexicographic maximization or also from a lexicographic minimization over the core. But, in both cases, all agents, firms and workers, take part in the optimization procedure.

In particular, given a market $(F,W,A,r)$,  if we take any order $\theta$ on $W$ and consider the extended order $\tilde{\theta}=(\underline{\theta}(1),\underline{\theta}(2),\ldots,\underline{\theta}(n))$,  the related max-min salary vector $y^{\tilde{\theta}}$ satisfies $y^{\tilde{\theta}}_{\theta(k)}\le y_{\theta(k)}$ for all $(x,y)\in C(v_{\gamma})$.  This is because each worker's payoff at $y^{\tilde{\theta}}$ tightly satisfies some lower core bound,  given the payoff of his/her predecessors.  As a consequence,  whenever $y^{\tilde{\theta}}$ belongs to the core,  it is the worst core allocation for workers,  and hence the vector of minimum competitive salaries, that supports the firm-optimal core allocation. Similarly, the worker-optimal core allocation follows from some $y^{\tilde{\theta}}$ where $\tilde{\theta}=(\overline{\theta}(1),\overline{\theta}(2),\ldots,\overline{\theta}(n))$. In the next section we discuss a special class of many-to-one assignment markets, where the max-min vector for both the all-min and the all-max extended orders always belong to the core, and can be directly obtained from the matrix without any computation.   



\section{Dominant diagonal markets}
\label{dominant_diagonal_markets}
A natural question arises after Corollary~\ref{corollary:FOCAWOCA}: 
for which many-to-one markets all workers attain a zero competitive salary and for which markets all workers $j$  attain $a_{j^\mu j}$ as a competitive salary? To this end we generalize the known condition for one-to-one assignment games due to \cite{sr01}.

\begin{definition}\label{def:dd}
A capacity-balanced many-to-one assignment market $(F,W,A,r)$ has a dominant diagonal if and only if there exists an optimal matching $\mu$ such that
\begin{enumerate}
\item $\sum_{j\in\mu(i)}a_{ij}\ge \sum_{j\in T}a_{ij}$ for all $T\subseteq W$ with $|T|\le r_i$, and 
\item $a_{j^\mu j}\ge a_{ij}$ for all $j\in W$ and $i\in F$.
\end{enumerate}
\end{definition}

Notice that when each worker $j\in W$ attains $a_{j^\mu j}$ as a competitive salary, then all firms have a zero payoff. The above property characterizes those many-to-one markets where each agent attains a zero payoff in the core.

\begin{proposition}
A capacity-balanced many-to-one assignment market $(F,W,A,r)$ has a dominant diagonal if and only if every agent attains a zero payoff in the core. 
\end{proposition}
\begin{proof}
Take an optimal matching $\mu$. If there is a core allocation where all workers $j\in W$ get salaries $y_j=0$, then, each firm $i\in F$ gets $x_i=\sum_{j\in\mu(i)}a_{ij}$. Now, such a payoff vector belongs to the core if and only if the core constraints for essential coalitions are satisfied, which means that for all $i\in F$, $\sum_{j\in\mu(i)}a_{ij}+0\ge \sum_{j\in T}a_{ij}$ for all $T\subseteq W$ with $|T|\le r_i$.

Similarly, when all firms get zero in a core allocation, then each worker gets the salary $y_j=a_{j^\mu j}$. This allocation belongs to the core if and only if for all $i\in F$,  $0+\sum_{j\in T}a_{j^\mu j}\ge \sum_{j\in T}a_{ij}$, for all $T\subseteq W$ with $|T|\le r_i$. This is equivalent to $a_{j^\mu j}\ge a_{ij}$ for all $j\in W$ and all $i\in F$.
\end{proof}

Notice that the market of Example \ref{example:core_workers} satisfies condition (2) of Definition \ref{def:dd}, but it does not satisfy condition (1), since for instance $\sum_{j\in\mu(2)}a_{ij}=a_{23}+a_{20}<a_{21}+a_{22}$. This is the reason why we observe in Figure \ref{fig:core} there is no core element where all workers receive zero salary.

It is important to remark that, since the above proposition shows that the dominant diagonal property is equivalent to a property of the core, which does not depend on the optimal matching we select for its representation, the diagonal dominant property is also independent of the optimal matching.

Let us also point out, without entering into formal details, that, as in the one-to-one assignment game, the dominant diagonal property is a necessary condition for the (von Neumann-Morgenstern) stability of the core. Indeed, if the point $(x,y)$, where $x_i=0$ for all $i\in F$ and $y_j=a_{j^\mu j}$ for all $j\in W$, is not in the core, then it cannot be dominated by any core allocation, which implies the core is not a stable set. The same happens if the point $(x,y)$, with $x_i=\sum_{j\in\mu(i)}a_{ij}$ for all $i\in F$ and $y_j=0$ for all $j\in W$, does not belong to the core.

On the other hand, in contrast to the one-to-one assignment game, the dominant diagonal property is not sufficient for stability of the core in the many-to-one assignment game.
Consider the following example. 
\begin{example}
\label{ex:dd_core_not_stable}
Let the market $\gamma=(F,W,B,r)$ be given by the pairwise valuation matrix
$$B =\bordermatrix{~ & w_{1} & w_{2} & w_{3} \cr
                  f_{1} & 6 & 4 & 1\cr
                  f_{2} & 5 & 4 & 5},$$
and firm-capacity vector $r=(2,1)$. In this capacity-balanced market, under the unique optimal matching $\mu$, workers $w_{1}$ and $w_{2}$ are matched with firm $f_{1}$, and $w_{3}$ with $f_{2}$. 
The market is clearly dominant diagonal, thus the firm-optimal core allocation is $(10,5;0,0,0)$ and the worker-optimal core allocation is $(0,0;6,4,5)$.

Take imputation $(x;y)=(0,5;6,4,0)$. It satisfies the  efficiency conditions $x_1+y_1+y_2=10$ and $x_2+y_3=5$, but it is not in the core because $x_1+y_3=0<1=v_{\gamma}(f_1,w_3)$.
We claim that no core allocation dominates this imputation.
Recall the general facts that (i) if an imputation dominates another imputation, then the first dominates the second via an essential coalition; and (ii) if a core allocation dominates an imputation, then it should be on the boundary of the core. 

Notice that imputation $(x;y)=(0,5;6,4,0)$ can only be dominated via $\{f_1,w_3\}$, $\{f_1,w_1,w_3\}$, or $\{f_1,w_2,w_3\}$ among the essential coalitions. However, a core allocation could dominate $(x;y)$ only via $\{f_1,w_3\}$, because the payoff of both $w_1$ and $w_2$ are at their core maximums $y_1=6$ and $y_2=4$.   
Suppose core allocation $(x';y')$ dominates imputation $(x;y)$ via $\{f_1,w_3\}$. Then it should be of the form 
$(x'_1=\varepsilon, x'_2=4+\varepsilon ;y'_1=6-\varepsilon_1,y'_2=4-\varepsilon_2,y'_3=1-\varepsilon)$ where $0<\varepsilon<1$, $0\leq\varepsilon_1,\varepsilon_2$  and $\varepsilon_1+\varepsilon_2=\varepsilon$.
Consider the core inequalities related to $\{f_1,w_1,w_3\}$ and $\{f_1,w_2,w_3\}$,

\qquad
$(x'_1=\varepsilon)+(y'_1=6-\varepsilon_1)+(y'_3=1-\varepsilon)=7-\varepsilon_1 \geq 7,$

\qquad
$(x'_1=\varepsilon)+(y'_2=4-\varepsilon_2)+(y'_3=1-\varepsilon)=5-\varepsilon_2 \geq 5.$
\\
Obviously, at least one of them is violated, because at least one of $\varepsilon_1>0$ or $\varepsilon_2>0$ holds.
We conclude that imputation $(x;y)=(0,5;6,4,0)$ can not be dominated by a core allocation in this dominant diagonal many-to-one assignment game.
\end{example}

Therefore, a characterization of core stability in terms of matrix properties requires further research.
Convexity of the game is a well-known sufficient condition.
The reader can easily show that this game property has the following matrix characterization.
\begin{proposition}
A many-to-one assignment game is convex if and only if no row or column of the underlying matrix contains more positive entries than the capacity of the corresponding agent. 
\end{proposition}

In contrast to the general case, in dominant diagonal many-to-one assignment games the tau-value is a core element, in fact, as in one-to-one assignment games, it is the midpoint of the line segment joining the two side-optimal core allocations.
\begin{proposition}
Let $(F,W,A,r)$ be a capacity-balanced many-to-one assignment market with dominant diagonal optimal matching $\mu$. 
Then
\begin{enumerate}
\item the payoff vector $(x^{f};y^{f})$ given by $x_i^{f}=\sum_{j\in\mu(i)}a_{ij}$ for all $i\in F$, $y_j^{f}=0$ for all $j\in W$, is a core allocation, (the \emph{firm-optimal core allocation}, where all firms get their marginal payoff to the grand coalition); 
\item the payoff vector $(x^{w};y^{w})$ given by $x_i^{w}=0$ for all $i\in F$, $y_j^{w}=a_{j^\mu j}$ for all $j\in W$, is a core allocation, (the \emph{worker-optimal core allocation}, where all workers get their marginal payoff to the grand coalition); 
\item the tau-value is a core allocation, the average of the two side-optimal core allocations, $\tau(v_{\gamma})=(x^{f}/2;y^{w}/2)\in\mathbf{C}(w_{\gamma})$.
\end{enumerate}
\end{proposition}
\begin{proof}
The first claim follows from the fact that if the optimal matching $\mu$ is dominant diagonal, then under $\mu$ any firm $i\in F$ is matched with her most profitable $r_i$ workers, that is $\min_{j\in\mu(i)} a_{ij} \geq \max_{k\in W\setminus\mu(i)} a_{ik}$. The second claim comes analogously, under $\mu$ any worker $j\in W$ is matched with his most profitable firm, that is $a_{j^\mu j}\geq \max_{i\in F} a_{ij}$.

It follows from the first two statements that the upper vector is $(x^{f};y^{w})$. It is easily seen that the lower  vector is $(x^{w}=0;y^{f}=0)$. Since $x^{f}(F)+y^{w}(W)=2v_{\gamma}(F\cup W)$, the efficiency scalar is $\kappa=1/2$.
Therefore, the tau-value is $\tau(w_A)=\frac{1}{2}(x^{f};y^{w})+\frac{1}{2}(x^{w};y^{f})=\frac{1}{2}(x^{f};y^{f})+\frac{1}{2}(x^{w};y^{w})=(x^{f}/2;y^{w}/2)$. By convexity of the core, $\tau(v_{\gamma})\in\mathbf{C}(v_{\gamma})$.
\end{proof}

Although in dominant diagonal many-to-one assignment games for all players their marginal contributions to the grand coalition are attained in the core, in fact simultaneously for both  sides, it is not true that all extreme core allocations are marginal vectors.
Consider Example~\ref{ex:dd_core_not_stable}, and take payoff vector $(5,4;\, 1,4,1)$, where no player receives 0 payoff. It is easily checked that it is an  extreme core allocation. Since the game is 0-normalized, for each order of the players the first player in the corresponding marginal vector gets 0 payoff. Thus, the given extreme core allocation cannot be a marginal vector.
It remains for future research to find out whether there is a larger subclass of many-to-one assignment games than the one-to-one assignment games where the COMA-property \citep{hetal02} holds.

\section{Kaneko's many-to-one buyer-seller market}\label{sec:kaneko}

The first many-to-one assignment game in the literature appears in \citep{k76}, as a market between buyers and sellers where each buyer demands only one unit while each seller may have several units on sale, even from different types. If we assume for simplicity that the goods owned by a seller are of the same type,  Kaneko's many-to-one assignment game is analogous to our job market assignment game from the perspective of the core and the theory of coalitional games. 

Let $B$ and $S$ be the finite and disjoint sets of buyers and sellers respectively, $A=(a_{ij})_{(i,j)\in B\times S}$ the pairwise valuation matrix and $r=(r_j)_{j\in S}$ the capacities of the sellers. Assume the market is capacity-balanced, that is $\sum_{j\in S}r_j=|B|$. By projecting the core of this game to the payoffs of the buyers (which is now the side with unitary capacity agents) analogously to Proposition \ref{core:description_workers}  we obtain that $(x,y)\in\mathbb{R}^{B}\times \mathbb{R}^{S}$ is in the core of the associated game $C(v_{\gamma})$, where $\gamma=(B,S,A,r)$, if and only if, for any optimal matching $\mu$, 
\begin{enumerate}[(i)]
\item $0\le x_{i}\leq a_{i\mu(i)}$ \, for any $i\in B$;
\item $x_k-x_i\ge a_{k\mu(i)}-a_{i\mu(i)}$ \, for any $i,k\in B$ such that $\mu(k)\ne \mu(i)$;
\item $y_{j}=\sum\limits_{i\in j^{\mu}}(a_{ij}-x_{i})$ for all $j\in S$.
\end{enumerate}

From this description of the core of Kaneko's assignment market, that we may call the buyers core $C(B)$, it follows the possibility of defining the  tight digraph associated with each core element. Now this graph at $x\in C(B)$ will have set of nodes $B$ and directed arcs related to those core inequalities  that are tight at $x$, in a way analogous to Definition \ref{def:tight_digraph}. As a consequence, we obtain a characterization of the extreme core allocations by means of the connectedness of its base-digraph, and characterizations of the buyers-optimal core element and the sellers-optimal core element parallel to those in Theorem \ref{thm:ext_core_char}: $x\in C(B)$ is the minimum core payoff vector for buyers if and only if its tight digraph contains a 0-sourced directed spanning tree, and it is the maximum core payoff vector for buyers if its tight digraph contains a 0-sinked directed spanning tree.

Also, a set of max-min payoff vectors $\{x^{\tilde{\theta}}\}_{\tilde{\theta}\in\tilde{\Sigma}}$ can be defined, one for each extended order on the set of buyers, and each extreme element of $C(B)$ is proved to be of this type, in a result parallel to Theorem \ref{thm:extremes are maxmin}.

However, regarding the set of competitive equilibrium payoff vectors, the two models clearly differ. \cite{k76} already shows by means of an example that although every competitive equilibrium payoff vector is in the core, not all core elements are supported by competitive prices. This is quite straightforward since in the above core description, two units from the same seller $j\in S$ that are sold to two different buyers $i,k\in B$ may have different price: $a_{ij}-x_i$ and $a_{kj}-x_k$. It is easy to see that the subset of core elements where the units of each seller are sold at the same price is the set of competitive equilibria payoff vectors.

\begin{proposition}\label{KanekoCE}
Let $\gamma=(B,S,A,r)$ be a capacity-balanced many-to-one assignment market where buyers have unitary capacity and $\mu$ an optimal matching. Then, $(x,y)\in\mathbb{R}^{B}\times \mathbb{R}^{S}$ is a competitive equilibrium payoff vector if and only if
\begin{enumerate}[(i)]
\item $0\le x_{i}\leq a_{i\mu(i)}$ \, for any $i\in B$;
\item $x_k-x_i\ge a_{k\mu(i)}-a_{i\mu(i)}$ \, for any $i,k\in B$;
\item $y_{j}=\sum\limits_{i\in j^{\mu}}(a_{ij}-x_{i})$ for all $j\in S$.
\end{enumerate}
\end{proposition}

Notice that the difference with the core, and with the CE equilibrium payoffs of our initial many-to-one job market, lies in the fact that inequality (ii) is required for each pair of buyers, not just for those that are not optimally matched to the same seller. This implies that if $i,k\in B$ are such that $\mu(i)=\mu(k)=j$, then (ii) gives $x_k-x_i=a_{k\mu(k)}-a_{i\mu(i)}$  which means that both units are sold at the same price: $p_j=a_{k\mu(k)}-x_k=a_{i\mu(i)}-x_i$.

\begin{example}\label{ex:Kaneko}
Consider the market $\gamma=(B,S,A,r)$ where the set of buyers is $B=\{b_1,b_2,b_3\}$, the set of sellers is $S=\{s_1,s_2\}$, the capacities of the sellers are $r=(2,1)$ and the valuation matrix is
$$A =\bordermatrix{~ & s_{1} & s_{2} \cr
                  b_{1} & 8 & 7 \cr
                  b_{2} & 6 & 6 \cr
                  b_{3} & 3 & 4 }.$$
 There is only one optimal matching $\mu=\{(b_1,s_1), (b_2,s_1), (b_3,s_3)\}$ and the core of this market consists of the set of payoff vectors $(x,y)\in\mathbb{R}^3\times\mathbb{R}^2$ such that 
 \begin{equation}\label{buyers core}\begin{array}{lll}
 0\le x_1\le 8  &  \quad 3\le x_1-x_3\le 5 & \quad y_1=(8-x_1)+(6-x_2)\\
 0\le x_2\le 6  &  \quad 2\le x_2-x_3\le 3 & \quad y_2=4-x_3\\
 0\le x_3\le 4 & & \end{array}\end{equation}       
                  
Notice that the valuation matrix is the transposed of Example \ref{example:core_workers}, and the capacities of sellers coincide with those of firms in that initial example. As a consequence notice that $C(B)$ coincides with $C(W)$ there. Hence, in our buyer-seller market, $(\overline{x}, \underline{y})=(8,6,4; 0,0)$ is the best core allocation for buyers while $(\underline{x},\overline{y})=(3,2,0; 9,4)$ is the best core allocation for sellers. However, in $(\underline{x},\overline{y})$, $s_1$ sells one unit to $b_1$ at the price $p_{1b_1}=a_{11}-x_1=5$ and sells a second unit to $b_2$ at the price $p_{1b_2}=a_{21}-x_2=4$, which means that $(3,2,0;9,4)$ is not supported by a competitive equilibrium.

To obtain the set of CE payoff vectors of this example, $CE(B)$, we only need to add to the set of inequalities (\ref{buyers core}) the fact that the two units of $s_1$ are sold at the same price,  $8-x_1=6-x_2$, that is $x_1-x_2=2$. By representing $CE(B)$, it is easy to check that it is the polytope spanned by the following four extreme vectors: $(4,2,0)$, $(5,3,0)$, $(8,6,3)$, and $(8,6,4)$. Then, the minimum CE payoff vector for the buyers is $(4,2,0)$, related to the CE prices $p_1=p_2=4$. 

Notice that the maximum payoff of the buyers in the core, $(8,6,4)$ satisfies the additional equation $x_1-x_2=2$ and hence it is supported by a competitive equilibrium and it is also the maximum CE payoff for buyers related with the minimum CE prices that are $p_1=p_2=0$.

\end{example}

We can provide a sufficient condition in terms of the pairwise valuation matrix that guarantees that all core allocations are supported by competitive prices.

\begin{proposition}
Let $(B,S,A,r)$ be a capacity-balanced many-to-one assignment market where buyers have unitary capacity, and $\mu$ an optimal matching. Then $C(B)=CE(B)$ if for all $j,j'\in S$ and $i,k,i'\in B$ such that $\mu(i)=\mu(k)=j$ and $\mu(i')=j'\neq j$ it holds
\begin{eqnarray}
a_{kj'}+a_{i'j}&\ge &a_{kj}+a_{i'j'},\label{eq1}\\
a_{ij'}+a_{i'j}&\ge &a_{ij}+a_{i'j'},\label{eq2}
\end{eqnarray}
\end{proposition}
\begin{proof}
Take $x\in C(B)$. From the core constraints,  together with (\ref{eq1}) and (\ref{eq2}), we get
\begin{eqnarray*}
x_k-x_i&=&(x_k-x_{i'})+(x_{i'}-x_i)\ge a_{kj'}-a_{i'j'}+a_{i'j}-a_{ij}\ge a_{kj}-a_{ij}, \mbox{ and }\\
x_i-x_k &=& (x_i-x_{i'})+(x_{i'}-x_k)\ge a_{ij'}-a_{i'j'}+a_{i'j}-a_{kj}\ge a_{ij}-a_{kj},\end{eqnarray*}
which proves that $x_k-x_i=a_{kj}-a_{ij}$. 
\end{proof}
Our previous characterization (Theorem \ref{thm:ext_core_char}) of the extreme competitive salaries of the multiple-partners job market can be straightforwardly extended to the extreme competitive buyers' payoffs of Kaneko's buyer-seller market, simply defining the tight digraph of a CE payoff vector using all the inequalities in Proposition \ref{KanekoCE}. Take, for instance, the extended tight digraph of the minimum  CE payoff vector for buyers in Example \ref{ex:Kaneko} (see Figure \ref{fig:Kaneko-minCE}) and notice that it has a unique source.

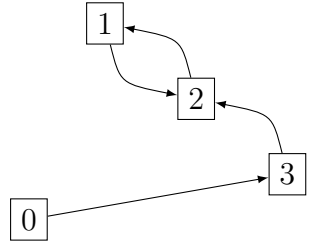
\begin{figure}[h]
\begin{center}
\begin{tikzpicture}[scale=1] 
\path 
	(0,-3.6) node(a00) [rectangle,draw] {$0$}
    (1,-1) node(a11) [rectangle,draw] {1}
    (2.2,-2) node(a22) [rectangle,draw] {2}
    (3.4,-3) node(a33) [rectangle,draw] {3}
 ; 
\draw[-latex] (a00) -- (a33) ;
\draw[-latex] (a11) .. controls +(+0.2,-0.8) .. (a22) ;
\draw[-latex] (a22) .. controls +(-0.2,+0.8) .. (a11) ;
\draw[-latex] (a33) .. controls +(-0.2,+0.8) .. (a22) ;
\end{tikzpicture}
\end{center}
\caption{Extended tight digraph of the minimum CE payoff vector (4,2,0).} \label{fig:Kaneko-minCE}
\end{figure}

Similarly, the definition of the max-min vectors in Definition \ref{def:maxmin} can be modified by the omission of the condition $j^{\mu}\ne j_r^{\mu}$. Then, a result analogous to Theorem \ref{thm:extremes are maxmin} guarantees that each extreme point of $CE(B)$ coincides with one of these max-min vectors.

\section{Concluding remarks}
\label{sec:remarks}
The core of many-to-one assingment markets had been studied to some extent, but little was known about its structure. In this paper, we have expanded the known results on (dis)similarities between the one-to-one case and the many-to-one case (and hence the many-to-many case). First, we have studied the relationship between the core and other solution concepts. We have shown the  kernel may not be included in the core, and remarkably, the coincidence between the core and the bargaining set does not hold. Our Example \ref{ex:kernel} represents a phenomenon that also happens in one-to-one markets: the short side of the market gets all the profit in any core allocation, without taking into account that the cooperation of the agents on the large side is needed to attain this profit. In the many-to-one markets, this situation may be solved by considering the kernel or the bargaining set to look for more fair distributions of the value of the market. In particular, in our example, the Shapley value is one of these kernel distributions that lie outside the core.

Secondly, the described procedure of the max-min salary vectors provides all extreme core allocations. Since all orders on the set of workers must be considered, this is not a very efficient procedure. Nevertheless, compared to the well-studied maximum and minimum competitive salary vectors, it allows to find other combinations of competitive salaries where the payoff of some workers is maximized while for others it is minimized, everything according to a given order of priority.

Let us finally point out that the negative results that have been presented, such as the non coincidence between core and bargaining set, or that not all extreme core points are lemaral vectors, trivially apply also to the class of many-to-many assignment markets.


\appendix
\section{Appendix}
\label{ex:not_INTO_ONTO}
This example shows that, for many-to-one assignment markets, neither all lemaral vectors are extreme core allocations nor all extreme core points can be obtained as lemaral vectors for some given order on the agents.
Consider the market $\gamma=(F,W,A,r)$ where $F=\{f_{1},f_{2}\}$, $W=\{w_{1},w_{2}\}$ are the set of firms and the set of workers respectively, and the capacities of the firms are $r=(2,1)$. The per-unit pairwise valuations are given in the following matrix:
 $$A =\bordermatrix{~ & w_{1} & w_{2} \cr
                  f_{1} & 4 & 3  \cr
                  f_{2} & 3 & 2 \cr}.$$
The corresponding many-to-one assignment game $(N,v_{\gamma})$ and its dual game are:
$$
\begin{array}{c|cccc|c}
v  & x_1 & x_2 & y_1 & y_2 & v^*\\
\hline
0 & 1 & . & . & . & 4 \\ 
0 & . & 1 & . & . & 0 \\ 
0 & . & . & 1 & . & 4 \\ 
0 & . & . & . & 1 & 3 \\ 
\hline
0 & 1 & 1 & . & . & 7 \\ 
4 & 1 & . & 1 & . & 5 \\ 
3 & 1 & . & . & 1 &  4 \\ 
3 & . & 1 & 1 & . & 4 \\ 
2 & . & 1 & . & 1 & 3 \\ 
0 & . & . & 1 & 1 & 7 \\ 
\hline
4  & 1 & 1 & 1 & . & 7 \\ 
3 & 1 & 1 & . & 1 & 7 \\ 
7 & 1 & . & 1 & 1 & 7 \\ 
3 & . & 1 & 1 & 1 & 7 \\ 
\hline
7 & 1 & 1 & 1 & 1 & 7 \\ 
\hline
\end{array} 
$$
The marginal payoff of firm $f_{1}$ is $4=v^{*}(f_{1})=v(N)-v(N\setminus f_{1})$ but it is not achievable in the core since the core-maximum for firm $f_{1}$ is $\max_{C}x_1=2=v^{*}(\{f_{2}\})+v^{*}(\{f_{1},w_{1}\})+v^{*}(\{f_{1},w_{2}\})-v^{*}(N)$. This shows that the marginal payoff of a player to the grand coalition may not be the core maximum payoff of the corresponding player for the many-to-one assignment game. 

Now, take any order that starts with the firm $f_{1}$, $\sigma=(f_{1}, arbitrary)$. For that given order, the payoff of $f_{1}$ is 4 which cannot be attained at a core allocation. Hence, a lemaral obtained by an order $\sigma=(f_{1}, arbitrary)$ cannot be a core allocation. 

Next, take the extreme core allocation $(2,0;3,2)$. Notice that $\min_{C}y_{2}=2=v(\{f_{2},w_{2}\})+v(\{f_{1},w_{1},w_{2}\})-v(N)$ and both $f_{1}$ and $f_{2}$ obtain their core maximum allocations, and hence $(2,0;3,2)$ is an extreme core allocation. We will try to construct a lemaral vector $(x,y)\in \mathbb{R}^{N}$ that coincide with the aforementioned extreme core allocation. First notice that $f_{2}$ is the only player that is paid her marginal payoff. Hence, we only take into account orders that start with player $f_{2}$:
\begin{itemize}
\item Player 2 achieves her marginal payoff under an order $\sigma=(f_{2}, arbitrary)$: $x_{2}=0$,
\item $\sigma=(f_{2},f_{1}, \ldots)$: Then, $$x_{1}=\min\{v^{*}(f_{1}),v^{*}(\{f_{1},f_{2}\})-x_{2}\}=\min\{4,7-0\}=4\neq 2=x_1,$$
\item $\sigma=(f_{2},w_{1}, \ldots)$: Then, $$y_{1}=\min\{4,4-0\}=4\neq 3=y_{1},$$
\item $\sigma=(f_{2},w_{2}, \ldots)$: Then, $$y_{2}=\min\{3,3-0\}=3\neq 2=y_{2}.$$
\end{itemize}
As a consequence, there does not exist an order to construct a lemaral vector that coincides with the extreme core allocation  $(2,0;3,2)$. 

 

\section{Appendix}
\label{ex:all_extended_payoffs}

All max-min salary vectors in the market from Example \ref{first_ex_revisited}:
\renewcommand{\arraystretch}{1.2}
$$
\begin{array}{|c|c@{\;}c@{\;}c|c|}
\textrm{ext. order} & y_1 & y_2 & y_3 & \textrm{in  core?}\\
\hline 
(\underline{1}, \underline{2}, \underline{3}) & 0 & 0 & 0 & - \\
(\underline{1}, \underline{2}, \overline{3}) & 0 & 0 & -3 & - \\
(\underline{1}, \overline{2}, \underline{3}) & 0 & 6 & 3 & - \\
(\underline{1}, \overline{2}, \overline{3}) & 0 & 6 & -3 & - \\
\hline 
(\overline{1}, \underline{2}, \underline{3}) & 8 & 0 & 3 & - \\
(\overline{1}, \underline{2}, \overline{3}) & 8 & 0 & -2 & - \\
(\overline{1}, \overline{2}, \underline{3}) & 8 & 6 & 3 & + \\
(\overline{1}, \overline{2}, \overline{3}) & 8 & 6 & 4 & + \\
\hline 
\end{array}
\qquad\qquad
\begin{array}{|c|c@{\;}c@{\;}c|c|}
\textrm{ext. order} & y_1 & y_2 & y_3 & \textrm{in  core?}\\
\hline 
(\underline{1}, \underline{3}, \underline{2}) & 0 & 2 & 0 & - \\
(\underline{1}, \underline{3}, \overline{2}) & 0 & 3 & 0 & - \\
(\underline{1}, \overline{3}, \underline{2}) & 0 & 5 & -3 & - \\
(\underline{1}, \overline{3}, \overline{2}) & 0 & 6 & -3 & - \\
\hline 
(\overline{1}, \underline{3}, \underline{2}) & 8 & 5 & 3 & + \\
(\overline{1}, \underline{3}, \overline{2}) & 8 & 6 & 3 & + \\
(\overline{1}, \overline{3}, \underline{2}) & 8 & 6 & 4 & + \\
(\overline{1}, \overline{3}, \overline{2}) & 8 & 6 & 4 & + \\
\hline 
\end{array}
$$
$$
\begin{array}{|c|c@{\;}c@{\;}c|c|}
\textrm{ext. order} & y_1 & y_2 & y_3 & \textrm{in  core?}\\
\hline 
(\underline{2}, \underline{1}, \underline{3}) & 0 & 0 & 0 & - \\
(\underline{2}, \underline{1}, \overline{3}) & 0 & 0 & -3 & - \\
(\underline{2}, \overline{1}, \underline{3}) & 8 & 0 & 3 & - \\
(\underline{2}, \overline{1}, \overline{3}) & 8 & 0 & -2 & - \\
\hline 
(\overline{2}, \underline{1}, \underline{3}) & 0 & 6 & 3 & - \\
(\overline{2}, \underline{1}, \overline{3}) & 0 & 6 & 4 & - \\
(\overline{2}, \overline{1}, \underline{3}) & 8 & 6 & 3 & + \\
(\overline{2}, \overline{1}, \overline{3}) & 8 & 6 & 4 & + \\
\hline 
\end{array}
\qquad\qquad
\begin{array}{|c|c@{\;}c@{\;}c|c|}
\textrm{ext. order} & y_1 & y_2 & y_3 & \textrm{in  core?}\\
\hline 
(\underline{2}, \underline{3}, \underline{1}) & 3 & 0 & 0 & - \\
(\underline{2}, \underline{3}, \overline{1}) & 5 & 0 & 0 & - \\
(\underline{2}, \overline{3}, \underline{1}) & 1 & 0 & -2 & - \\
(\underline{2}, \overline{3}, \overline{1}) & 3 & 0 & -2 & - \\
\hline 
(\overline{2}, \underline{3}, \underline{1}) & 6 & 6 & 3 & + \\
(\overline{2}, \underline{3}, \overline{1}) & 8 & 6 & 3 & + \\
(\overline{2}, \overline{3}, \underline{1}) & 7 & 6 & 4 & + \\
(\overline{2}, \overline{3}, \overline{1}) & 8 & 6 & 4 & + \\
\hline 
\end{array}
$$
$$
\begin{array}{|c|c@{\;}c@{\;}c|c|}
\textrm{ext. order} & y_1 & y_2 & y_3 & \textrm{in  core?}\\
\hline 
(\underline{3}, \underline{1}, \underline{2}) & 3 & 2 & 0 & + \\
(\underline{3}, \underline{1}, \overline{2}) & 3 & 3 & 0 & + \\
(\underline{3}, \overline{1}, \underline{2}) & 5 & 2 & 0 & + \\
(\underline{3}, \overline{1}, \overline{2}) & 5 & 3 & 0 & + \\
\hline 
(\overline{3}, \underline{1}, \underline{2}) & 7 & 6 & 4 & + \\
(\overline{3}, \underline{1}, \overline{2}) & 7 & 6 & 4 & + \\
(\overline{3}, \overline{1}, \underline{2}) & 8 & 6 & 4 & + \\
(\overline{3}, \overline{1}, \overline{2}) & 8 & 6 & 4 & + \\
\hline 
\end{array}
\qquad\qquad
\begin{array}{|c|c@{\;}c@{\;}c|c|}
\textrm{ext. order} & y_1 & y_2 & y_3 & \textrm{in  core?}\\
\hline 
(\underline{3}, \underline{2}, \underline{1}) & 3 & 2 & 0 & + \\
(\underline{3}, \underline{2}, \overline{1}) & 5 & 2 & 0 & + \\
(\underline{3}, \overline{2}, \underline{1}) & 3 & 3 & 0 & + \\
(\underline{3}, \overline{2}, \overline{1}) & 5 & 3 & 0 & + \\
\hline 
(\overline{3}, \underline{2}, \underline{1}) & 7 & 6 & 4 & + \\
(\overline{3}, \underline{2}, \overline{1}) & 8 & 6 & 4 & + \\
(\overline{3}, \overline{2}, \underline{1}) & 7 & 6 & 4 & + \\
(\overline{3}, \overline{2}, \overline{1}) & 8 & 6 & 4 & + \\
\hline 
\end{array}
$$
\renewcommand{\arraystretch}{1}

Notice that since workers 1 and 2 are optimally matched to the same firm, thus the difference between their core payoffs is not constrained, whenever they occupy consecutive positions in an extended order the associated max-min vector is the same. Based on this observation, the full enumeration process can be somewhat shortened.

\bibliography{many2one}

\begin{thebibliography}{39}
\newcommand{\enquote}[1]{``#1''}
\providecommand{\natexlab}[1]{#1}
\providecommand{\url}[1]{\texttt{#1}}
\providecommand{\urlprefix}{URL }
\providecommand{\bibAnnoteFile}[1]{%
  \IfFileExists{#1}{\begin{quotation}\noindent\textsc{Key:} #1\\
  \textsc{Annotation:}\ \input{#1}\end{quotation}}{}}
\providecommand{\bibAnnote}[2]{%
  \begin{quotation}\noindent\textsc{Key:} #1\\
  \textsc{Annotation:}\ #2\end{quotation}}

\bibitem[{Atay and Solymosi(2018)}]{as18}
Atay, A. and T.~Solymosi (2018), \enquote{On bargaining sets of
  supplier-firm-buyer games.} \emph{Economics Letters}, 167, 99--103.
\bibAnnoteFile{as18}

\bibitem[{Bahel(2016)}]{b16}
Bahel, E. (2016), \enquote{On the core and bargaining set of a veto game.}
  \emph{International Journal of Game Theory}, 45, 543--566.
\bibAnnoteFile{b16}

\bibitem[{Bahel(2021)}]{b21}
Bahel, E. (2021), \enquote{Hyperadditive games and applications to networks or
  matching problems.} \emph{Journal of Economic Theory}, 191, 105168.
\bibAnnoteFile{b21}

\bibitem[{Ba{\"\i}ou and Balinski(2000)}]{bb00}
Ba{\"\i}ou, M. and M.~Balinski (2000), \enquote{The stable admissions
  polytope.} \emph{Mathematical Programming}, 87, 427--439.
\bibAnnoteFile{bb00}

\bibitem[{Ba{\"\i}ou and Balinski(2002)}]{bb02}
Ba{\"\i}ou, M. and M.~Balinski (2002), \enquote{The stable allocation (or
  ordinal transportation) problem.} \emph{Mathematics of Operations Research},
  27, 485--503.
\bibAnnoteFile{bb02}

\bibitem[{Balinski and Gale(1990)}]{bg90}
Balinski, M.~L. and D.~Gale (1990), \enquote{On the core of the assignment
  game.} In \emph{Functional Analysis, Optimization and Mathematical
  Economics}, 274--289, Oxford University Press.
\bibAnnoteFile{bg90}

\bibitem[{Benedek et~al.(2021)Benedek, Fliege, and Nguyen}]{betal21}
Benedek, M., J.~Fliege, and T.-D. Nguyen (2021), \enquote{Finding and verifying
  the nucleolus of cooperative games.} \emph{Mathematical Programming}, 190,
  135--170.
\bibAnnoteFile{betal21}

\bibitem[{Davis and Maschler(1965)}]{dm65}
Davis, M. and M.~Maschler (1965), \enquote{The kernel of a cooperative game.}
  \emph{Naval Research Logistics Quarterly}, 12, 223--259.
\bibAnnoteFile{dm65}

\bibitem[{Davis and Maschler(1967)}]{dm67}
Davis, M. and M.~Maschler (1967), \enquote{Existence of stable payoff
  configurations for cooperative games.} In \emph{Essays in Game Theory and
  Mathematical Economics in Honor of Oskar Morgenstern} (M.~Shubik, ed.),
  39--52, Princeton University Press.
\bibAnnoteFile{dm67}

\bibitem[{Demange(1982)}]{d82}
Demange, G. (1982), \enquote{{S}trategyproofness in the assignment market
  game.} \emph{Laboratorie d'{\'E}conom{\'e}trie de l'{\'E}cole Polytechnique,
  Mimeo, Paris}.
\bibAnnoteFile{d82}

\bibitem[{Dom\`{e}nech and N\'{u}\~{n}ez(2022)}]{dn22}
Dom\`{e}nech, G. and M.~N\'{u}\~{n}ez (2022), \enquote{Axioms for the optimal
  stable rules and the fair division rules in a multiple-partners job market.}
  \emph{Games and Economic Behavior}, 136, 469--484.
\bibAnnoteFile{dn22}

\bibitem[{Granot and Granot(1992)}]{gg92}
Granot, D. and F.~Granot (1992), \enquote{On some network flow games.}
  \emph{Mathematics of Operations Research}, 17, 792--841.
\bibAnnoteFile{gg92}

\bibitem[{Gul and Stachetti(1999)}]{gs99}
Gul, F. and E.~Stachetti (1999), \enquote{Walrasian equilibrium with gross
  substitutes.} \emph{Journal of Ecomomics Theory}, 87, 95--124.
\bibAnnoteFile{gs99}

\bibitem[{Hamers et~al.(2002)Hamers, Klijn, Solymosi, Tijs, and
  Villar}]{hetal02}
Hamers, H., F.~Klijn, T.~Solymosi, S.~Tijs, and J.~P. Villar (2002),
  \enquote{Assignment games satisfy the {C}o{M}a-property.} \emph{Games and
  Economic Behavior}, 38, 231--239.
\bibAnnoteFile{hetal02}

\bibitem[{Izquierdo et~al.(2007)Izquierdo, N{\'u}{\~n}ez, and Rafels}]{ietal07}
Izquierdo, J.~M., M.~N{\'u}{\~n}ez, and C.~Rafels (2007), \enquote{A simple
  procedure to obtain the extreme core allocations of an assignment market.}
  \emph{International Journal of Game Theory}, 36, 17--26.
\bibAnnoteFile{ietal07}

\bibitem[{Kaneko(1976)}]{k76}
Kaneko, Mamoru (1976), \enquote{On the core and competitive equilibria of a
  market with indivisible goods.} \emph{Naval Research Logistics Quarterly},
  23, 321--337.
\bibAnnoteFile{k76}

\bibitem[{Leonard(1983)}]{l83}
Leonard, H.~B. (1983), \enquote{{E}licitation of honest preferences for the
  assignment of individuals to positions.} \emph{Journal of Political Economy},
  91, 461--479.
\bibAnnoteFile{l83}

\bibitem[{N\'{u}\~{n}ez and Rafels(2002)}]{nr02}
N\'{u}\~{n}ez, M. and C.~Rafels (2002), \enquote{{T}he assignment game: the
  $\tau$-value.} \emph{International Journal of Game Theory}, 31, 411--422.
\bibAnnoteFile{nr02}

\bibitem[{N{\'u}{\~n}ez and Rafels(2015)}]{nr15}
N{\'u}{\~n}ez, M. and C.~Rafels (2015), \enquote{A survey on assignment
  markets.} \emph{Journal of Dynamics and Games}, 2, 227--256.
\bibAnnoteFile{nr15}

\bibitem[{N{\'u}{\~n}ez and Solymosi(2017)}]{ns17}
N{\'u}{\~n}ez, M. and T.~Solymosi (2017), \enquote{Lexicographic allocations
  and extreme core payoffs: the case of assignment games.} \emph{Annals of
  Operations Research}, 254, 211--234.
\bibAnnoteFile{ns17}

\bibitem[{Peleg(1966)}]{p66}
Peleg, B. (1966), \enquote{The kernel of the general-sum four-person game.}
  \emph{Canadian Journal of Mathematics}, 18, 673--677.
\bibAnnoteFile{p66}

\bibitem[{P\'{e}rez-Castrillo and Sotomayor(2017)}]{pcs17}
P\'{e}rez-Castrillo, D. and M.~Sotomayor (2017), \enquote{On the manipulability
  of competitive equilibrium rules in many-to-many buyer-seller markets.}
  \emph{International Journal of Game Theory}, 46, 1137--1161.
\bibAnnoteFile{pcs17}

\bibitem[{Quint(1991)}]{q91}
Quint, T. (1991), \enquote{Characterization of cores of assignment games.}
  \emph{International Journal of Game Theory}, 19, 413--420.
\bibAnnoteFile{q91}

\bibitem[{S\'{a}nchez-Soriano et~al.(2001)S\'{a}nchez-Soriano, Lopez, and
  Garcia-Jurado}]{ssetal01}
S\'{a}nchez-Soriano, J., M.~A. Lopez, and I.~Garcia-Jurado (2001),
  \enquote{{O}n the core of transportation games.} \emph{Mathematical Social
  Sciences}, 41, 215--225.
\bibAnnoteFile{ssetal01}

\bibitem[{Schmeidler(1969)}]{s69}
Schmeidler, D. (1969), \enquote{{T}he nucleolus of a characteristic function
  game.} \emph{SIAM Journal of Applied Mathematics}, 17, 1163--1170.
\bibAnnoteFile{s69}

\bibitem[{Shapley and Shubik(1971)}]{ss71}
Shapley, L. and M.~Shubik (1971), \enquote{{T}he assignment game {I}: {T}he
  core.} \emph{International Journal of Game Theory}, 1, 111--130.
\bibAnnoteFile{ss71}

\bibitem[{Shapley(1953)}]{s53}
Shapley, L.S. (1953), \enquote{A value for n-person games.} In
  \emph{Contributions to the theory of games II}, 307--317, Princeton
  University Press.
\bibAnnoteFile{s53}

\bibitem[{Solymosi(1999)}]{s99}
Solymosi, T. (1999), \enquote{On the bargaining set, kernel and core of
  superadditive games.} \emph{International Journal of Game Theory}, 28,
  229--240.
\bibAnnoteFile{s99}

\bibitem[{Solymosi(2002)}]{soly02}
Solymosi, T. (2002), \enquote{The bargaining set of four-person balanced
  games.} \emph{International Journal of Game Theory}, 31, 1--11.
\bibAnnoteFile{soly02}

\bibitem[{Solymosi(2008)}]{sol08}
Solymosi, T. (2008), \enquote{Bargaining sets and the core in partitioning
  games.} \emph{Central European Journal of Operations Research}, 16, 425--440.
\bibAnnoteFile{sol08}

\bibitem[{Solymosi and Raghavan(2001)}]{sr01}
Solymosi, T. and T.E.S. Raghavan (2001), \enquote{Assignment games with stable
  core.} \emph{International Journal of Game Theory}, 30, 177--185.
\bibAnnoteFile{sr01}

\bibitem[{Solymosi et~al.(2003)Solymosi, Raghavan, and Tijs}]{setal03}
Solymosi, T., T.E.S. Raghavan, and S.~Tijs (2003), \enquote{Bargaining sets and
  the core in permutation.} \emph{Central European Journal of Operations
  Fesearch}, 11, 93--101.
\bibAnnoteFile{setal03}

\bibitem[{Sotomayor(1992)}]{s92}
Sotomayor, M. (1992), \enquote{The multiple partners game.} In
  \emph{Equilibrium and Dynamics}, 322--354, Springer.
\bibAnnoteFile{s92}

\bibitem[{Sotomayor(2002)}]{so02}
Sotomayor, M. (2002), \enquote{{A} labor market with heterogeneous firms and
  workers.} \emph{International Journal of Game Theory}, 31, 269--283.
\bibAnnoteFile{so02}

\bibitem[{Sotomayor(2007)}]{s07}
Sotomayor, M. (2007), \enquote{Connecting the cooperative and competitive
  structures of the multiple-partners assignment game.} \emph{Journal of
  Economic Theory}, 134, 155--174.
\bibAnnoteFile{s07}

\bibitem[{Thompson(1981)}]{t81}
Thompson, G.~L. (1981), \enquote{{A}uctions and market games.} In \emph{Essays
  in Game Theory and Mathematical Economics in Honor of Oskar Morgenstern}
  (R.~Aumann, ed.), 181--196, Bibliographisches Institut, Mannheim.
\bibAnnoteFile{t81}

\bibitem[{Tijs(1981)}]{tijs81}
Tijs, S. (1981), \enquote{{B}ounds for the core of a game and the
  $\tau$-value.} In \emph{Game Theory and Mathematical Economics} (O.~Moeschlin
  and D.~Pallaschke, eds.), 123--132, North-Holland Publishing Company,
  Amsterdam.
\bibAnnoteFile{tijs81}

\bibitem[{Trudeau and Vidal-Puga(2017)}]{tvp17}
Trudeau, C. and J.~Vidal-Puga (2017), \enquote{On the set of extreme core
  allocations for minimal cost spanning tree problems.} \emph{Journal of
  Economic Theory}, 169, 425--452.
\bibAnnoteFile{tvp17}

\bibitem[{van~den Brink et~al.(2021)van~den Brink, N{\'u}{\~n}ez, and
  Robles}]{vbnr21}
van~den Brink, R., M.~N{\'u}{\~n}ez, and F.~Robles (2021), \enquote{Valuation
  monotonicity, fairness and stability in assignment problems.} \emph{Journal
  of Economic Theory}, 195, 105277.
\bibAnnoteFile{vbnr21}

\end{thebibliography}
\bibliographystyle{te}
\end{document}